\documentclass[conference]{IEEEtran}
\makeatletter
\def\ps@headings{%
\def\@oddhead{\mbox{}\scriptsize\rightmark \hfil \thepage}%
\def\@evenhead{\scriptsize\thepage \hfil \leftmark\mbox{}}%
\def\@oddfoot{}%
\def\@evenfoot{}}
\makeatother
\usepackage{amsmath}
\usepackage{subfigure}
\usepackage{amsfonts}
\usepackage{graphicx}
\usepackage{epstopdf}
\usepackage{verbatim}
\usepackage{setspace}
\usepackage[usenames]{color}
\usepackage{multirow,url}
\usepackage{cite,url}
\usepackage{algorithmicx}
\usepackage{algpseudocode}
\usepackage{amsthm}

\allowdisplaybreaks

\newtheorem{theorem}{Theorem}

\newtheorem{lemma}{Lemma}
\newtheorem{definition}{Definition}

\ifodd 0
\newcommand{\rev}[1]{{\color{blue}#1}} %revise of the text
\newcommand{\com}[1]{\textbf{\color{red}(COMMENT: #1)}} %comment of the text
\newcommand{\clar}[1]{\textbf{\color{green}(NEED CLARIFICATION: #1)}}

\else
\newcommand{\rev}[1]{#1}
\newcommand{\com}[1]{}
\newcommand{\clar}[1]{}
\fi

\begin{document}

\title{Profit Incentive In A Secondary Spectrum Market: A Contract Design Approach}
\author{ Shang-Pin Sheng,  Mingyan Liu \\
Electrical Engineering and Computer Science, University of Michigan\\ 
\{shangpin, mingyan\}@umich.edu}
\maketitle

\begin{abstract}
In this paper we formulate a contract design problem where a primary license holder wishes to profit from its excess spectrum capacity by selling it to potential secondary users/buyers.  It needs to determine how to optimally price the excess spectrum so as to maximize its profit, knowing that this excess capacity is stochastic in nature, does not come with exclusive access, and cannot provide deterministic service guarantees to a buyer.   At the same time, buyers are of different {\em types}, characterized by different communication needs, tolerance for the channel uncertainty, and so on, all of which a buyer's private information.  The license holder must then try to design different contracts catered to different types of buyers in order to maximize its profit.  We address this problem by adopting as a reference a traditional spectrum market where the buyer can purchase exclusive access with fixed/deterministic guarantees.  
We fully characterize the optimal solution in the cases where there is a single buyer type, and when multiple types of buyers share the same, known channel condition as a result of the primary user activity.  In the %case when each type of buyers experience different channel conditions 
most general case we construct a algorithm that generates a set of contracts in a computationally efficient manner, and show that this set is optimal when the buyer types satisfy a monotonicity condition.
%We consider two cases; in one the seller has full information on the buyer, including its service requirement and quality constraint, and in the other the buyer only knows a set of types the buyer may belong to and the associated probabilities. In the first case we fully characterize the nature of the optimal contract design, while in the case we do so for the case of two types.  For more than two types we illustrate a procedure that can be shown to be optimal under certain conditions.  
\end{abstract}

%% INTRODUCTION %%
\section{Introduction}

%There is an enormous growth on the demand for wireless resources recently, resulted both from the increased number of wireless devices and the increased bandwidth needed for each individual device. The resource (spectrum) is limited in the sense that the range (of frequencies) that can be used to transmit data is limited and there are restrictions imposed by the government on each frequency band. Any device that uses the wireless spectrum is compeled to follow the Federal Communications Commission’s (FCC) rules in the United States. The right to use this wireless resource has to be obtained either by purchasing an exclusive license on a particular bandwidth or to use an unlicensed band of the spectrum, where exclusive licenses have to be obtained through participating in auctions held by the FCC. Since July 1994, the FCC has raised over \$60 billion for the U.S. over spectrum auction which speaks for it self how valuable these spectrum resources are. As a consequence, only specific groups with strong financial power are able to participate and win bandwidths in these auctions.

The scarcity of spectrum resources and the desire to improve spectrum efficiency have  led to extensive research and development in recent years in such concepts as dynamic spectrum access/sharing, open access, and secondary (spot or short-term) spectrum market, see e.g., \cite{akyildiz2006next,buddhikot2007understanding}. 
From the inception of the open access paradigm, it was clear that for it to work two issues must be adequately addressed: sensing and pricing. The first refers to the ability of a (secondary) device to accurately detect channel opportunity and more generally to acquire information on the spectrum environment. %, so as not to interfere with a primary user thereby protecting the latter, and to opportunistically utilize spectrum availability. 
The second refers to mechanisms that provide license holders with the right incentives so that they will willingly allow access by unlicensed devices.

There has been a number of mechanisms proposed to address this incentive issue, the most often used being the auction mechanism, see e.g., \cite{wu2008multi,gandhi2007general,huang2006auction}.  Auction is also the primary mechanism used in allocating spectrum on the primary market \cite{grimm2003low}.  In this paper we consider an alternative approach, that based on {\em contracts}, to the trading of spectrum access on the secondary market (see Section \ref{sec:discussion} on a discussion comparing the two mechanisms). 
%Under an auction, competing buyers submit bids to a license holder to obtain spectrum access.  In selecting winning bids an auction can be designed to maximize the profit of the license holder \cite{gandhi2007general}, or to maximize social welfare \cite{huang2006auction}, or some combination of both \cite{gandhi2007general}.  In contrast, under a contract framework the focus is less on the competition among buyers, but more on the design of a set of contracts on the part of the license holder to attract potential buyers.  
%
This is conceptually not unlike the design of pricing plans by a cellular operator: it presents a potential user with a set of contract options, each consisting of parameters including the duration of the contract, discount on the device, number of free minutes per month, price per minute for those over the free limit, window of unlimited calling time, and so on.  In coming up with these calling plans the operator typically studies carefully the types of callers it wants to attract and their calling patterns/habits; the subsequent plans are catered to these patterns with the objective of maximizing its revenue.  A caller interested in entering into contract with the operator is expected to look through these plans and pick one that is the best suited for him/her needs.  

In this paper we adopt such a contract design approach in the context of the secondary spectrum market, where a license holder advertises a set of prices and service plans in the hope that a potential buyer will find one of them sufficiently appealing to enter into contract.  
%In adopting a contract design approach, we aim to address two fundamental issues in this paper.  The first is the incentive issue for the license holder: 
The contracts are designed with the goal of maximizing the expected revenue of the license holder given a set of buyer {\em types} (more precisely defined in the next section). 

To make the contracts appealing to a buyer, one must address the issue that the spectrum  offered on the secondary (short-term) market is %often uncertain in nature. % and therefore it must be priced appropriately.
%
%Specifically, one of the fundamental premises behind a secondary/short-term spectrum market is the existence of excess capacity due to the primary license holder's own spectrum under-utilization.  However, 
%Specifically, spectrum offered on the secondary market is 
typically the excess capacity due to the primary license holder's own spectrum under-utilization.  Its quality is therefore often uncontrolled and random, both spatially and temporally, and strongly dependent on the behavior of the primary users.  %One may be able to collect statistics and make predictions, as has been done in numerous spectrum usage studies \cite{mchenry2006cso,mchenry2005,liu-mobicom09}, but it is fundamentally stochastic in nature.  
The primary license holder can of course choose to eliminate the randomness by setting aside resources (e.g., bandwidth) exclusively for secondary users.  This will however likely impinge on its current users and may not be in the interest of its primary business model.  
%, dynamic spectrum sharing introduces a good way of increasing the efficiency on utilizing spectrum resources \cite{akyildiz2006next}, \cite{buddhikot2007understanding}. One cooperative scheme of dynamic spectrum sharing  is through data relaying. In such cases, a primary user picks several secondary users to relay data transmission for him. Since the secondarys are closer to the receiver, higher transmission rates can be acheived. In return, he pays the secondary users money or allows the secondary users to use the channel for their own goods for a fixed amount of time \cite{wang2010cooperative}, \cite{zhang2009stackelberg}. This particular sharing scheme can be formulated as a contract problem for the primary user, which becomes more interesting and harder when the primary user does not know exactly the type of the secondary users. Work on this type of contract problems are shown in \cite{duan2011contract} .
%
The alternative is to simply give non-exclusive access to secondary users for a fee, which allows the secondary users to {\em share} a certain amount of bandwidth simultaneously with its existing licensed users, but only under certain conditions on the primary traffic/spectrum usage.  For instance, a secondary user is given access but can only use the bandwidth if the current activity by the licensed users is below a certain level, e.g., as measured by received SNR, the so-called spectrum overlay.  %This is a typical scenario under the spectrum overlay and underlay models \cite{zhao-survey07}; 
Many spectrum sharing schemes proposed in the literature fall under this scenario, see e.g., \cite{kim08,xin-liu06,zhao07,ahmad-tit08}. 

In this case a secondary user pays (either in the form of money or services in return) to gain spectrum access but not for guaranteed use of the spectrum.  This presents a challenge to both the primary and the secondary users: On one hand, the secondary user must assess its needs and determine whether the uncertainty in spectrum quality is worth the price asked for and what level of uncertainty can be tolerated.  On the other hand, the primary must decide how stochastic service quality should be priced so as to remain competitive against guaranteed (or deterministic) services which the secondary user may be able to purchase from a traditional market or a different primary license holder.  
To address this challenge we adopt a reference point in the form of a traditional spectrum market from where a secondary user can purchase guaranteed service, i.e., exclusive access rights to certain bandwidth, at a fixed price per unit.  %assume there is a fair spectrum market that offers the secondary user a purchase of the exclusive right to transmit at some frequency. 
This makes it possible for the secondary user to reject the offer from the primary if it is risk-averse or if the primary's offer is not attractive. This also implies that the price per unit of bandwidth offered by the primary user must reflect its stochastic quality. %be lower than the market price. 

Work most relevant to the study presented in this paper includes \cite{duan2011contract,muthuswamyportfolio}. In \cite{duan2011contract} a contract problem is studied where the secondary users help relay primary user's data and in return are allowed to send their own data.  In  \cite{muthuswamyportfolio} an optimal portfolio problem is studied, where a secondary user can purchase a bundle of different stochastic channels, with the price of each already determined, and seeks to find the optimal purchase.  %By contrast, in this paper we seek to optimize the seller's profit, while noting that the buyer also behaves in self-interest. 
%, under two types of bandwidth shortage constraints. The work in \cite{muthuswamyportfolio} however considers 
%, assuming there is already a market that determines a fixed price for each type of stochastic channel. They considered the problem from 
%only the perspective of the buyer but not the seller.  
%\com{we probably need to strengthen this part} 
%but did not take into account the behavior of the spectrum seller.} 

%In this paper, we assume the secondary's have a demand satisfaction rate constraint (as in \cite{muthuswamyportfolio}). We tackle the problem by proposing a contract-based cooperative spectrum sharing mechanism \cite{bolton2005contract}.  This problem becomes for the primary user to design a good contract in order to attract the secondary user to cooperate with him. On the same time, to achieve the maximum reward on selling his product (excess bandwidth).

%TODO%
Main contributions of this paper are as follows:
\begin{enumerate}
\item We formulate a contract design problem where the spectrum license holder seeks to sell his excess bandwidth to potential buyers. The model captures the following essential features: (1) excess bandwidth on the secondary spectrum market often comes with non-exclusive use and therefore highly uncertain channel conditions; (2) incentives are built in for both the seller and the buyer to conduct spectrum trading on the secondary market. %the buyer to buy the transmission right on channels that do not guarantee throughput.
\item We fully characterize the optimal set of contracts the seller should provide in the case of a single or two types of buyers, and when multiple types of buyers share the same channel condition due to primary user activities. 
%\item When there are two possible types of buyer, we derive the optimal set of contracts.
\item When there are multiple types of buyers and each experiences different channel conditions, we construct a computationally efficient algorithm and show that the set of contracts it generates is optimal when the buyer types satisfy a monotonicity condition.
%When there can be more than two possible types of buyer, who  When the buyer types share the common channel condition, we constructed an algorithm that finds the optimal solution. When the buyer types have different channel condition, we constructed an algorithm that finds the optimal solution under the monotonicity condition.
\end{enumerate}

The remainder of the paper is organized as follows. We present the contract design problem in Section \ref{sec:model}. Section \ref{sec:single-type} characterizes the utility region and the optimal contract in the single buyer case. Section \ref{sec:public} deals with the case when the channel condition is common knowledge, while Section \ref{sec:private} focuses on the case when channel conditions are private knowledge.  Discussion is given in Section \ref{sec:discussion} and numerical results in Section \ref{sec:simulation}.
\section{Model and Assumptions}\label{sec:model} 

In this section we describe in detail the models for the two parties under the contract framework: the seller and the buyer, and their considerations in designing and accepting a contract, respectively.   We also illustrate a basic idea underlying our model to capture the value of secondary spectrum service, which is random and non-guaranteed in nature, by using guaranteed service as a reference.

%\subsection{The Seller}

%\subsection{The contract design model} 

There are two parties to a contract, the seller and the buyer.  The {\em seller} is also referred to as the owner or the primary license holder, who uses the spectrum to provide business and service to its {\em primary users}, and carry {\em primary traffic}. 
He is willing to sell whatever underutilized bandwidth he has as long as it generates positive profit and does not impact negatively his primary business. 
We will assume that the seller can pre-design up to $M$ contracts and announce them to potential buyers. If a buyer accepts one of the contracts, they come to an agreement and they have to follow the contract up to a predetermined period of time. %It is up to the seller to design the contracts, and the buyer to decide whether or not to accept it. 
We will leave this duration unspecified as it does not affect our analysis under the current model. 
%\subsection{The contract} 

Each contract is in the form of a pair of real numbers $(x, p)$, where $x\in R^+$ and $p\in R^+$.
\begin{itemize}
\item $x$ is the amount of bandwidth they agree to trade on (i.e., access to this amount of bandwidth is given from the seller to buyer).
\item $p$ is the price per unit of $x$; thus a total of $xp$ is paid to the seller if the buyer purchases this contract.  \end{itemize}

%When a contract $(x, p)$ is signed, 
The seller's profit or utility from contract $(x, p)$ is given as 
\begin{equation*}
U(x,p)=x(p-c)
\end{equation*}
where $c$ is a predetermined constant that takes into account the operating cost of the seller. 
We will assume that any contract the seller presents must be such that $p>c$; that is, the seller will not sell at a loss. If none of the contracts is accepted by the buyer, the {\em reserve utility} of the owner is defined by $U(0,0)=0$.  

%\subsection{A reference market of fixed/deterministic service or exclusive use}

%\com{First, all results presented are based on a monopoly market operated by a single spectrum seller. No market competition is considered in the paper. This is, however, a very restrictive model as a typical spectrum market contains multiple sellers competing with each other, and the contract offered by one spectrum seller is inevitably affected by others.}

%\com{Besides, the authors also assume discrete buyer types, i.e., values of q, x, and p are discrete so that the total number of types is limited.}

{We next consider what a contract specified by the pair $(x, p)$ means to a potential buyer.  To see this, we will assume that there exists a traditional (as opposed to this emerging, secondary) market from where the buyer can purchase services with fixed or deterministic guarantees.  What this means is that the buyer can purchase {\em exclusive} use of certain amount of bandwidth, which does not have to be shared with other (primary) users.  This serves as an alternative to the buyer, and is used in our model as a point of reference.  We will not specify how the price of exclusive use is set, and will simply normalize it to be unit price per unit of bandwidth (or per unit of transmission rate). The idea is that given this alternative, the seller  cannot arbitrarily set his price because the buyer can always walk away and purchase from this traditional market.  This traditional market will also be referred to as the {\em reference} market, and the service it offers as the {\em fixed} or {\em deterministic} service.  Our model allows a buyer to purchase from both markets should that be in the interest of the buyer.  Note that even though we have assumed a single seller model, this is not a monopoly because of the existence of this reference market.  However, we do not explicitly model the competition between multiple sellers on the secondary market, which remains an interesting subject of future study.
} 

%\subsection{Market Price}
%
%The market price is served as an alternative for the buyer. We assume that the market only sells bandwidth in the tranditional way, where each unit of bandwidth is sold at a fixed unit price but is guaranteed to let the buyer transmit at the rate of channel capacity. We are not interested in how this price should be set, so we normalize this price as unit price. Given this market price, the owner cannot arbitrarily set his price unreasonably high because the buyer can always reject to work with the owner and go with this market price.
%\begin{itemize}
%\item There is an existing spectrum market, 1 unit of money get 1 unit of bandwidth. (deterministic tranmission rate)
%\end{itemize}

%\subsection{The buyer's consideration} 

When the set of $M$ contracts are presented to a buyer, his choices are (1) to choose one of the contracts and abide by its terms, (2) to reject all contracts and go to the traditional market, and (3) to purchase a certain combination from both markets.  The buyer's goal is to minimize his purchasing cost as long as certain quality constraints are satisfied.  

While the framework presented here applies to any meaningful quality constraint, to make our discussion concrete below we will focus on a loss constraint.  
Suppose the buyer chooses to purchase $y$ units of fixed service from the reference market together with a contract $(x, p)$.  Then its constraint on expected loss of transmission can be expressed as: 
\begin{equation*}
E[(q-y-x B)^+]\leq \epsilon ~, 
\mbox{where} 
\end{equation*}
\begin{itemize}
\item $q$ is the amount of data/traffic the buyer wishes to transmit.
\item $B\in \{0, 1\}$ is a binary random variable denoting the quality of the channel for this buyer.  We will denote $b:=P(B=1)$. 
%\item $b$: the probability of the channel (from the owner) generating transmission rate.
\item $\epsilon$ is a threshold on the expected loss acceptable to the buyer. 
%the amount of data that can be accepted to be lost (Expected loss).
\end{itemize}
%
%\rev{There are a number of ways to interpret the random variable $B$ and the amount $xB$. 
Note that quantities $x, y$ and $q$ are of the same unit; this unit can be bit (total amount of transmission), or rate (bits per second), and so on. 
%
%\com{he characterization of the stochastic nature of the excessive available spectrum is too simple, and doesn't seem to take into account spatio-temporal dynamics in available bandwidth as the number of users change.}
%
Here we have adopted a simplifying assumption that the purchased service (in the amount of $x$) is either available in the full amount (when $B=1$) or unavailable (when $B=0$), with $xb$ being the expected availability.  
If the contract duration is comparable to the time constant of the primary user activity (e.g., peak vs. off-peak hours) then this model captures the spectrum condition at the time of contract signing. 
More sophisticated models can be adopted here, by replacing $x B$ with another random variable $X(x)$ denoting the random amount of data transmission the buyer can actually realize.  
%This will not affect the framework presented here, but may alter the technical details. 
%A stronger assumption is made here that the purchase of the channel either acheives full capacity or generates zero transmission rate. We denote here $B\in\{0,1\}$ as a binary random variable that describes the outcome of the secondary channel. Under this assumption, the data rate generated from this channel is $x$ with probability $b$ and $0$ with probability $1-b$.

With a purchase of $(y, (x, p))$, the buyer's cost is given by 
%\begin{equation*}
%C_b(x,p) = 
$y+xp$. 
%\end{equation*}
%\begin{itemize}
%\item $y$: the amount of bandwidth he purchases from the market.
%\end{itemize} 
%where $(x,p)$ is the contract they agree on. 
%If the buyer does not enter into any of the presented contract, we assume his cost is simplified to $C_b(0,0)=y$, i.e., he will only purchase from the traditional market to fulfill his needs. %However, the purchase $y$ will in general be larger than the $y$ when they have a contract deal.
%Now we will introduce the quality constraint (loss constraint) model used in this paper. 
%The buyer is assumed to have a predetermined constraint on his transmitted data which is in the form of expected loss on transmission. The constraint is such that the expected loss $E[loss]$ has to be less than some value $\epsilon$. More precisely,
%\begin{equation*}
%E[(q-y-xB)^+]<\epsilon
%\end{equation*}
%\begin{itemize}
%\item $q$: the amount of data the buyer wishes to transmit.
%\item $b$: the probability of the channel (from the owner) generating transmission rate.
%\item $\epsilon$: the amout of data that can be accepted to be lost (Expected loss).
%\end{itemize}
%A stronger assumption is made here that the purchase of the channel either acheives full capacity or generates zero transmission rate. We denote here $B\in\{0,1\}$ as a binary random variable that describes the outcome of the secondary channel. Under this assumption, the data rate generated from this channel is $x$ with probability $b$ and $0$ with probability $1-b$.
The cost of the contract $(x, p)$ to this buyer is given by the value of the following minimization problem: 
%We can now formulate the buyer's utility as an optimization problem.
\begin{eqnarray}
C(x,p)= \underset{y}{\text{minimize}} & & y+xp \\
  \text{subject to} & & E[(q-y-xB)^+]\leq \epsilon \label{constraint}
\end{eqnarray}
That is, to assess how much this contract actually costs him, the buyer has to consider how much additional fixed service he needs to purchase to fulfill his needs. 

The buyer can always choose to not enter into any of the presented contracts and only purchase from the traditional market.  In this case, his cost is given by the value of the following minimization problem: 
\begin{eqnarray*}
C(0,0)= \underset{y}{\text{minimize}} ~ y,   %\\
~~  \text{subject to} ~ E[(q-y)^+]\leq \epsilon
\end{eqnarray*}
Since every term is deterministic in the above problem, we immediately conclude that $C(0,0)=q-\epsilon$, which will be referred to as the \emph{reserve price} of the buyer. 
It is natural to assume that any buyer must be such that $q \geq \epsilon$, for otherwise the buyer does not need to perform any transmission as it can tolerate the loss of all of its data.
 
%Obviously if a contract's cost is higher than this price then there is no incentive for the buyer to enter into that contract. 

In deciding whether to accept a given contract $(x, p)$, the buyer has to consider (1) whether the contract would satisfy its quality (loss) constraint, and (2) whether there is an incentive to enter into this contract, i.e., whether the cost of this contract is no higher than the reserve price.  The latter is also referred to as the {\em individual rationality} (IR) constraint, %more precisely given as 
%
%Subject to its constraint, the buyer accepts the contract if it has an incentive to do so, referred to as the \emph{individual rationality} (IR) constraint. In other words, the buyer (by accepting the contract) has to be able to achieve a cost no higher than the reserve price:  
$C(x,p)\leq C(0,0)=q-\epsilon$.  Any contract that satisfies both constraints of a buyer is referred to as {\em acceptable} to that buyer.

%\subsection{Buyer types and informational constraints}

We will assume that a potential buyer may be one of a number of different {\em types}; each type is characterized by a unique triple $(q, \epsilon, b)$, which is a buyer's {\em private information}.  That is, a type is characterized by its transmission needs (amount $q$ to be transferred and loss requirement $\epsilon$), as well as its perceived spectrum/channel quality ($b$).  
Throughout the paper we will assume that a type $(q, \epsilon, b)$ is such that there exists a contract with $p>c$ acceptable to the buyer, for otherwise the seller has no incentive to sell.  %anything to this buyer and can thus remove it from its consideration. 

We will further assume two cases, where $b$ is common to all types and where $b$ may be different for different types.  The first case models the scenario where buyers are relatively homogeneous and their perceived channel quality is largely determined by the primary user traffic reflected in $b$.  In this case it is also natural to assume that $b$ is known to the seller.  The second case models the scenario where buyers may differ significantly in their location, quality of transceiver devices, and so on, which leads to different perceived channel quality, which is only known to a buyer himself. 

The seller is assumed to know the distribution of the types but not the actual type of a particular buyer.  Specifically, we will assume there are $K$ types of buyers, and a buyer is of type $i$ with probability $r_i$ and is given by the triple $(q_i, b_i, \epsilon_i)$.  
Continuous type distribution is discussed in Section \ref{sec:discussion}. 
%We will also assume that at most $M$ different contracts are announced to the buyer. 
 %
In subsequent sections we proceed in the following sequence:  
%\com{needs to revise after the sections are done...} 
% the following three cases each with increasing complexity. 
%\begin{enumerate}
%\item 
(1) single user type,   %We will start with this simplest case.  Since we assume the seller is aware of the a priori type distribution, in this case the seller effectively knows the triple $(q, \epsilon, b)$.  %We will show in this case the seller can extract all of the buyer's surplus (over the reserve price), resulting in $C(x,p)=C(0,0)$ at the optimal contract point.
%This simplest case is used to establish basic understanding of the problem and introduce a few key concepts and tools used in more complex scenarios. 
%
%\item 
(2) multiple user types; common $b$, %  This is the case where there are $K$ types, but all of them share the same $b$ which is known to the seller.  For this case we derive the set of optimal contracts the seller should use. 
and (3) 
%\item 
multiple user types; different and private $b$.  \rev{Due to space limit, we omit some proofs and offer intuitive explanations instead. Complete proofs can be found in \cite{fullversion}.}
\section{Optimal contract for a single buyer type} \label{sec:single-type} 

We begin by considering the case where there is only one type of buyer $(q, \epsilon, b)$.  Through this simplified scenario we will introduce a number \rev{of} concepts key to our analysis and obtain some basic understanding of the nature of this problem. 

Under our assumption that the seller knows the buyer type distribution, having a single type (i.e., a singleton distribution) essentially means that the triple $(q, \epsilon, b)$ is known to the seller.   %The seller can thus custom-design a contract for the buyer. 
%
%In considering whether to accept this contract, the buyer 
%Subject to its constraint, the buyer accepts the contract if it has an incentive to do so, referred to as the \emph{individual rationality} (IR) constraint. In other words, the buyer (by accepting the contract) has to be able to achieve a cost no higher than the reserve price:  
%$C_b(x,p)\leq C_b(0,0)=q-\epsilon$. 
%
%Specifically, 
Denote by $T = \{(x,p): C(x,p)\leq C(0,0)\}$ the set of all acceptable contracts for the buyer, or the {\em acceptance region}.   
%The seller can precisely determine this region with the knowledge $(q, \epsilon, b)$. %, this region where the buyer would accept a contract. % $(x, p)$. (
%This is illustrated in Fig. \ref{fig:boundary}. 
%The shape of this acceptance region 
This is characterized by the next result. 
\begin{theorem}\label{thm:acceptance_region}
When $q(1-b)\leq\epsilon$, the buyer accepts a contract $(x, p)$ iff 
\begin{eqnarray}
p&\leq&\left \{ \begin{tabular}{c c}
$b$ & if $x\leq\frac{q-\epsilon}{b}$\label{eqn:Thm1-1}\\
$\frac{q-\epsilon}{x}$ & if $x>\frac{q-\epsilon}{b}$
\end{tabular}\right. ~. \label{eqn:acceptance_region_1}
\end{eqnarray}
When $q(1-b)>\epsilon$, the buyer accepts the contract iff  
\begin{eqnarray}
p&\leq&\left \{ \begin{tabular}{c c}
$b$ & if $x\leq\frac{\epsilon}{1-b}$\label{eqn:Thm1-2}\\
$\frac{b\epsilon}{x(1-b)}$ & if $x>\frac{\epsilon}{1-b}$
\end{tabular}\right. ~. \label{eqn:acceptance_region_2} 
\end{eqnarray}
\end{theorem}
%\rev{The above theorem can be proved for each of the cases listed above by showing the IR and the loss constraint is satisfied/not satisfied.  For brevity we ommit the proofs here.}
%\begin{comment}
The above theorem can be proved for each of the cases listed above.  For brevity below we only show the proof for the sufficient condition under $q(1-b)\leq\epsilon$ for the first case in  Eqn (\ref{eqn:acceptance_region_1}); other cases can be done using similar arguments.
% proof %
\begin{lemma}\label{lem:acceptance_region_1}
When $q(1-b)\leq\epsilon$, the buyer accepts the contract $(x,p)$ if $x\leq \frac{q-\epsilon}{b}$ and $p\leq b$. 
\end{lemma}
\begin{proof}
If both the IR constraint and the loss constraint are satisfied under the stated conditions, then the buyer accepts the contract.  Below we check these two constraints. 
Let the buyer supplement this contract with an additional purchase of $y=q-\epsilon-xp$ deterministic service.  Note that $y\geq 0$ under the stated conditions. The total cost of this contract to the buyer is then given by:
%We start by letting $y=q-\epsilon-xp$ and show that the IR constraint is satisfied.
\begin{equation*}
C(x, p) = y+xp=q-\epsilon-xp+xp=q-\epsilon = C(0,0) .
\end{equation*}
The IR constraint is therefore satisfied.  The buyer's loss under this combination of purchases is given by: 
\begin{eqnarray*}
&E&[(q-y-xB)^+]\\
&=&(q-y-x)^+b+(q-y)^+(1-b)\\
&=&(\epsilon+xp-x)^+b+(\epsilon+xp)(1-b)\\
&=&\left \{
\begin{array}{l}
(\epsilon+xp)(1-b)\leq(\epsilon+b\frac{q-\epsilon}{b})(1-b)  \\
~~~ =q(1-b)\leq \epsilon, ~~~~~~~~~  \mbox{if } \epsilon+x(p-1)\leq 0\\
(\epsilon+x(p-1))b+(\epsilon+xp)(1-b)  \\
~~~ =\epsilon+x(p-b)\leq\epsilon, ~~~~ \mbox{if } \epsilon+x(p-1)> 0
\end{array}\right.
\end{eqnarray*}
Thus the loss constraint is also satisfied. 
%$\epsilon+x(p-1)\leq 0$ in the upper case and  $\epsilon+x(p-1)> 0$ in the lower case.
\end{proof}
%\end{comment}
%\begin{lemma}\label{lem:acceptance_region_2}
%When $q(1-b)\leq\epsilon$, the buyer accepts the contract $(x,p)$ if $x\geq \frac{q-\epsilon}{b}$ and $xp\leq q-\epsilon$. %= C_b(0,0)$, the reserve price. 
%\end{lemma}
%\begin{proof}
%Let the buyer purchase no deterministic service: $y=0$.  The IR constraint is immediately satisfied due to the condition $xp\leq q-\epsilon$. 
%We now verify the loss constraint: 
%\begin{eqnarray*}
%E[(q-xB)^+]&=&(q-x)^+b+q(1-b)\\
%&\leq&(q-\frac{q-\epsilon}{b})^+b+q(1-b)\\
%&=&(qb-(q-\epsilon))^++q(1-b)\\
%%&=&(\epsilon-q(1-b))^++q(1-b))\\
%&=&(\epsilon-q(1-b))+q(1-b))=\epsilon, 
%\end{eqnarray*}
%where the second to last equality follows from the fact that $q(1-b)\leq \epsilon$.
%\end{proof}
The two acceptance regions given by Theorem \ref{thm:acceptance_region} are illustrated in Figs. \ref{fig:boundary}.  
%\com{to save space I suggest we put the two curves in the same figure but label them separately rather than in the figure title...} 
Any contract that falls below the boundary is acceptable to the buyer.  The two cases have the following interpretation.  
In the first case when $q(1-b)\leq \epsilon$, % (or $b\geq \frac{q-\epsilon}{q}$), 
the quality of the stochastic channel is sufficiently good such that the loss constraint (\ref{constraint}) may be met %(when $x\geq \frac{q-\epsilon}{b}$) 
without any purchase of the deterministic channel. % (this is demonstrated in Lemma \ref{lem:acceptance_region_2}).  
In this case the buyer is willing to spend up to the entire reserve price $C(0,0)=q-\epsilon$ on the contract. %: from Eqn (\ref{eqn:acceptance_region_1}) we see $xp \leq q-\epsilon$ regardless of $x$. 
In the second case when $q(1-b)> \epsilon$, % (or $b< \frac{q-\epsilon}{q}$), 
the quality of the stochastic channel is such that no matter how much is purchased, some deterministic channel is needed ($y>0$) to satisfy the loss constraint %: %Thus the buyer is not willing to spend all of $q-\epsilon$ on the contract.  
(note $xp \leq \frac{b\epsilon}{1-b} < q-\epsilon$ because $q(1-b)> \epsilon$).
\begin{figure}[h]
\centering
\includegraphics[width=0.4\textwidth]{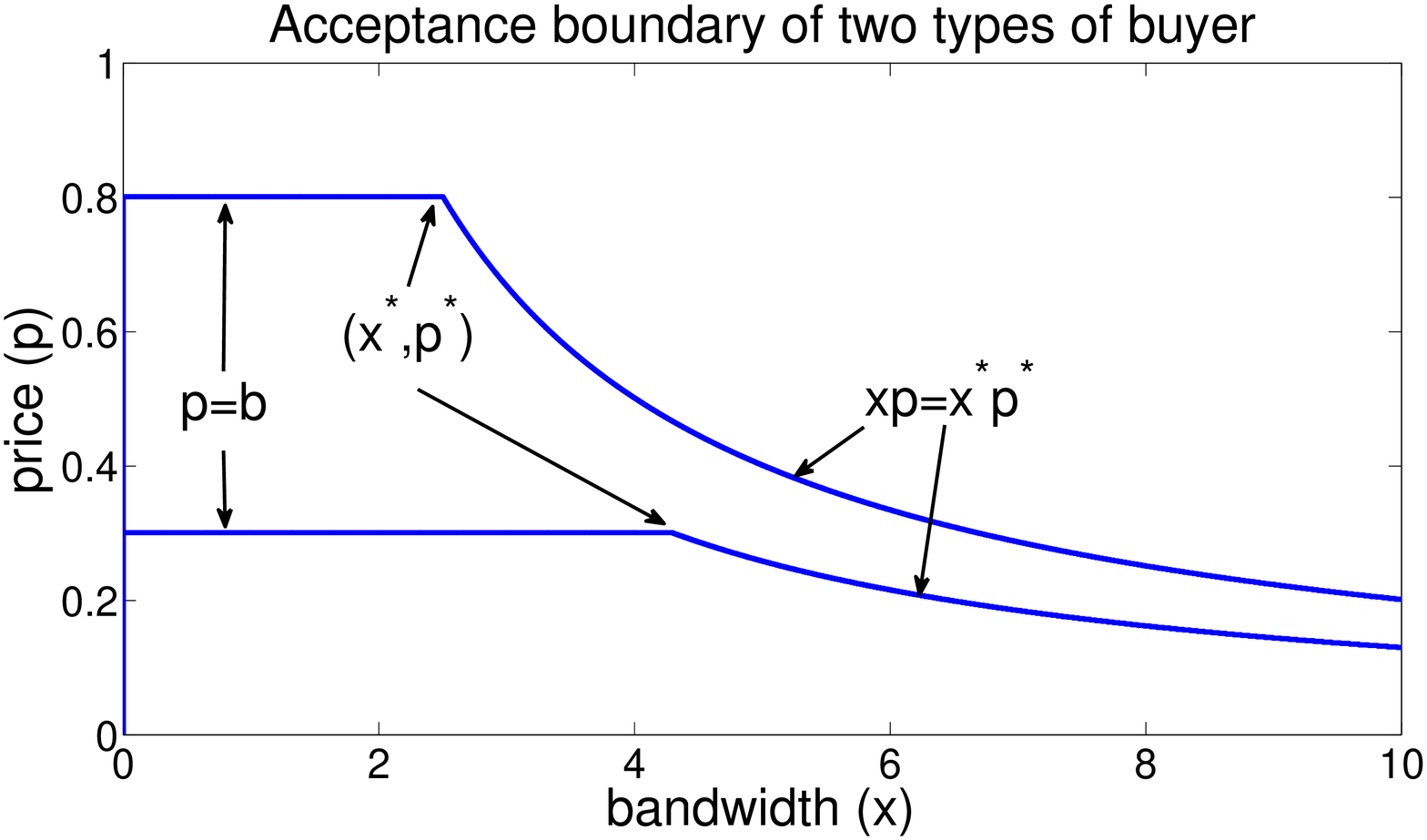}
\caption{The upper curve is when $q(1-b)< \epsilon(q=5, b=0.8,\epsilon = 3)$, the lower curve is when $q(1-b)< \epsilon(q=5,b=0.3,\epsilon=3)$}
\label{fig:boundary}
\end{figure}
%\begin{figure}[h]
 % \centering
  %     \subfigure[Case of $q(1-b)\leq \epsilon$]{
     %     \label{fig:boundary_1}
   % \includegraphics[width=0.45\textwidth]{fig1_writeup.eps}}
   % \subfigure[{Case of $q(1-b)> \epsilon$}]{
    %\label{fig:boundary_2}
   % \includegraphics[width=0.45\textwidth]{fig2_writeup.eps}} 
 % \caption{Examples of acceptance regions}
%\label{fig:boundary}
%\end{figure}
%Let $(x^*,p^*)$ denote the ``knee'' (the intersection point where the straight line meets the curve) on the boundary of the acceptance region. 
%The lemma below gives a compact expression for this point.
%\begin{lemma}
%The point $(x^*,p^*)$ on the boundary of the acceptance region of a buyer $(q, \epsilon, b)$ is given by $(x^*,p^*)=(min(\frac{\epsilon}{1-b}, \frac{q-\epsilon}{b}), b)$. % regardless of the conditions.
%\end{lemma}
%\begin{proof}
%Using Theorem \ref{thm:acceptance_region}, when $q(1-b)\leq \epsilon$,
%\begin{eqnarray*}
%%x^*&=&
%\frac{q-\epsilon}{b} &\leq&\frac{\frac{\epsilon}{1-b}-\epsilon}{b} 
%=\frac{\epsilon}{1-b} \\
%\Rightarrow ~~~ x^* &=& \frac{q-\epsilon}{b} = min(\frac{\epsilon}{1-b}, \frac{q-\epsilon}{b})
%\end{eqnarray*}
%Similarly, when $q(1-b)> \epsilon$,
%\begin{eqnarray*}
%\frac{q-\epsilon}{b}&>&\frac{\frac{\epsilon}{1-b}-\epsilon}{b} %\\
%=\frac{\epsilon}{1-b}\\
%\Rightarrow ~~~ x^*&=&\frac{\epsilon}{1-b}=min(\frac{\epsilon}{1-b}, \frac{q-\epsilon}{b})
%\end{eqnarray*}
%\end{proof}
%
%After determining the feasible region of contracts 
For a given buyer type ($q,\epsilon,b$), the seller can choose any point in the corresponding acceptance region $T$ to maximize its utility: 
$\max_{(x,p) \in T}  U(x,p)$. 
%. The optimal contract is the one that returns the highest payoff, denoted as $(x^*,p^*)$: % and defined as follows:
%\begin{eqnarray*}
%(x^*,p^*)=\underset{(x,p) \in T}{\text{argmax}} & U_o(x,p)
%\end{eqnarray*}
We next show that the optimal contract for the seller is given by the ``knee'' (the intersection point where the straight line meets the curve) on the boundary of the acceptance region, 
denoted as $(x^*,p^*)$. 
\begin{theorem}
The optimal contract for the seller is the intersection point $(x^*, p^*)$ on the acceptance region boundary of the buyer.
\end{theorem}

%The theorem can be proved in two steps. First show that the seller's utility is strictly increasing in $p$ which implies that the optimal contract must be such that (\ref{eqn:Thm1-1}) and (\ref{eqn:Thm1-2}) hold with strict equality. Then show that the intersection point is strictly better than any other point on the boundary.
%\begin{comment}
\begin{proof}
We prove the optimality in two steps. First we show that the seller's utility is strictly increasing in $p$ which implies that the optimal contract must be such that (\ref{eqn:Thm1-1}) and (\ref{eqn:Thm1-2}) hold with strict equality. Then we show that the intersection point is strictly better than any other point on the boundary.
%\begin{enumerate}
%\item 
For any $x>0$ and $\forall p'>p$, we have 
\begin{eqnarray*}
U(x,p')=x(p'-c) %=xp'-xc
>x(p-c)=U(x,p). 
\end{eqnarray*}
Thus $U(x,p)$ is strictly increasing in $p$.
%\item 
%Next we check the boundary.  
%
For any $x<x^*$ (points on the straight line) we have % and \rev{given $p^*>c$}, we have 
\begin{eqnarray*}
U(x^*,p^*)=x^*(p^*-c) > x(p^*-c)=U(x,p^*), 
\end{eqnarray*}
which used the fact that $p^*>c$. (Recall we have assumed that for any buyer there must exist a contract with $p>c$ that it finds acceptable.  This implies such a point must be within the acceptance region, which in turn implies that we must have $p^*>c$ since $p^*\geq p$, $\forall p$ in the region.)
For any pair $(x,p)$ such that $xp=x^*p^*$ and $x>x^*$ (points on the curve),
\begin{eqnarray*}
U(x,p) = x(p-c)%=xp-xc\\
=x^*p^*-xc >  x^*(p^*-c)=U(x^*,p^*). 
\end{eqnarray*}
Thus $U(x^*,p^*)$ is strictly greater than any point $U(x,p)$ on the boundary. 
%\end{enumerate}
\end{proof}
%\end{comment}
Once the seller determines the optimal contract and presents it to the buyer, the buyer will accept because it satisfies both the loss and the IR constraints. It can be easily shown that the buyer's cost in accepting is exactly $C(0,0)$. 
Note that technically since the cost of the contract is exactly equal to the reserve price, the buyer is ambivalent between getting only deterministic service and getting a mix of both types of services.  In practice the seller can always lower the unit price $p^*$ by an arbitrarily small amount to provide a positive incentive so that the buyer will accept the contract. %The key is to identify this point of ambivalence so the seller knows how to induce the desired outcome.  
For this reason even though the costs are equal, for simplicity we will assume that the buyer will accept this contract.
For the same reason, we will also assume that when there exist multiple contracts of equal cost to the buyer, the seller can always induce the desired choice from the buyer by introducing a small difference to the desired contract. %induce the buyer pick the one the owner assigns.
%\com{need to consider whether this is the best place for this sentence...} 
We have now a complete characterization of the contract design for a single type of buyer.  %Before we move onto the case of multiple buyer types in the next two sections, 
We end this section by introducing the concept of an {\em equal-cost line} of a buyer.  Consider a contract $(x', p')$. Denote by $P(x',p',x)$ a price such that 
%:= Given a contract $(x',p')\in T_i$ and $x$ representing the amount of bandwidth sold for some contract, this function defines the price such that 
the contract $(x,P(x',p'$ $,x))$ has the same cost as contract $(x',p')$ to a buyer. %of type $(q, \epsilon, b)$. 
%This will be referred to as an {\em equivalent price}.  Obviously $P(x',p',x)$ is a function of $x$, $x'$, and $p'$. % a type $i$ buyer.
\begin{definition}
The equal-cost line $E$ of a buyer of type $(q, \epsilon, b)$ is the set of contracts within the buyer's acceptance region $T$ that are of equal cost to the buyer. Thus $(x, p) \in E$ if and only if $p=P(x', p', x)$ for some other $(x', p')\in E$. The cost of this line is given by $C(x', p')$, $\forall (x', p')\in E$.
\end{definition}
It should be clear that there are many equal-cost lines, each with a different cost.  
Figure \ref{fig:equalcostline} shows an example of a set of equal-cost lines. 
We will therefore also write an equal-cost line as $E_{x', p'}$ for some $(x', p')$ on the line to distinguish it from other equal-cost lines. The next theorem gives a precise expression for the equivalent price that characterizes an equal-cost line.
\begin{figure}
\centering
\includegraphics[width=0.4\textwidth]{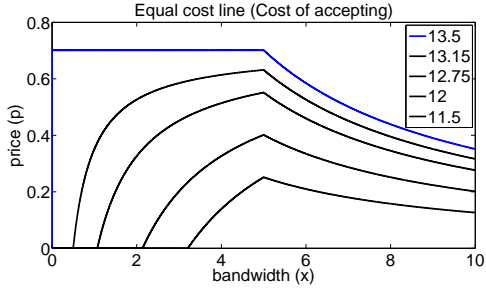}
\caption{Example of equal cost lines}
\label{fig:equalcostline}
\vspace{-20pt}
\end{figure}
%\com{I revised the wording in the theorem below because $P()$ is not a set unless we given a region for $(x', p')$, and even then it's a set of equivalent prices, not a set of contracts, which the line actually represents. However, I'm not sure this is the best way to go about it.  We may still want to define directly this line... }  
\begin{theorem}\label{thm:single_buyer}
For a buyer of type $(q, \epsilon, b)$ with an intersection point $(x^*, p^*)$ on its acceptance region boundary, and given a contract $(x', p')$, an equal-cost line $E_{x', p'}$ consists of all contracts $(x, P(x',p',x))$ such that
\begin{eqnarray*}
P(x',p',x)=
\left\{\begin{tabular}{ll}
$b-\frac{x'}{x}(b-p')$ & if $x,x'\leq x^*$\\
$x'p'/x$ & if $x,x'\geq x^*$\\
$(b(x^*-x')+x'p')/x$ & if $x'<x^*<x$\\
$b-(x^*b-x'p')/x$ & if $x<x^*<x'$
\end{tabular}\right.
\end{eqnarray*}\label{Thm:equal_cost}
\end{theorem}
\vspace{-10pt}
\begin{proof}
We will prove this for the case $q(1-b)\leq \epsilon$; the other case can be shown with similar arguments and is thus omitted for brevity. In this case $x^*= \frac{q-\epsilon}{b}$.
When $x,x'\leq x^*$, without buying deterministic service the loss is given by
\begin{eqnarray*}
E[(q-xB)^+]&=&(q-x)^+b+q(1-b) \\
&=& (q-x)b+q(1-b) %\\
= q-xb \geq \epsilon, 
\end{eqnarray*}
where the second equality is due to the fact that $q(1-b)\leq \epsilon \Rightarrow \frac{q-\epsilon}{b} \leq q \Rightarrow x\leq \frac{q-\epsilon}{b}\leq q$.
The incentive for the buyer is to purchase $y$ such that the loss is just equal to $\epsilon$.
\begin{eqnarray*}
E[(q-y-xB)^+]&=&(q-y-x)b+(q-y)(1-b)\\
&=&q-y-xb=\epsilon~. 
\end{eqnarray*}
The first equality follows from the fact that $q(1-b)\leq \epsilon$, which implies both $(q-y-x)\geq 0$ and $(q-y)\geq 0$. This is true for both $(x,p)$ and $(x',p')$. Since $(x,p)$ is on the equal cost line $E_{x',p'}$, we know that $C(x,p)=C(x',p')$. We also know that $C(x,p)=y+xp$ and $C(x',p')=y'+x'p'$,
\begin{eqnarray*}
C(x,p)=q-\epsilon-xb+xp=q-\epsilon-x'b+x'p'=C(x',p')~. 
\end{eqnarray*}
Rearranging the second equality such that $p$  is a function of $x,x',p'$ immediately gives the result. When $x,x' > x^*$, $x$ ($x'$) alone is sufficient to achieve the loss constraint. For $C(x,p)=C(x',p')$ we must have $x'p'=xp$, resulting in the second branch. The third and fourth branch can be directly derived from the first two branches. When $x>x^*>x'$ ($x'>x^*<x$), we first find the equivalent price at $x^*$ by the first branch (second branch), and then use the second branch (first branch) to find $P(x',p',x)$. This gives the third branch (fourth branch)\qed
\end{proof}
\rev{The form of the equal-cost line is the same regardless whether $q(1-b)\leq \epsilon$ or $q(1-b)> \epsilon$.} Note that every contract below an equal-cost line is strictly preferable to a contract on the line for the buyer.  This is an observation we will use in subsequent sections. We end this section with a property of the equivalent price we will use later.
\begin{lemma}
$P(x',p',x)$ is strictly increasing in $p'$ when $x'> 0$. \label{Lemma:strict_increasing}
\end{lemma}
\vspace{-10pt}
This lemma is easily shown by noting $C(x',p')=y+x'p'$, where $y$ is only a function of $x'$. Thus, $p>p'$ implies $C(x',p)>C(x',p')$ when $x'>0$.
%\begin{figure}
%\centering
%\includegraphics[width=0.5\textwidth]{equalcostline.eps}
%\caption{Example of equal cost lines}
%\label{fig:equalcostline}
%\end{figure}
%
%\begin{figure}
%\centering
%\includegraphics[scale=0.25]{fig5_writeup.eps}
%\caption{The regions to distinguish type-1 given $(x^2,p^2)$}
%\label{fig:equal_utility}
%\end{figure}
%\begin{figure}
%\centering
%\includegraphics[scale=0.25]{fig6_writeup.eps}
%\caption{An example of 3 possible types of buyer}
%\label{fig:n_type}
%\end{figure}

\section{Common channel condition} \label{sec:public}

%Starting from this section, we will look at the case when there are multiple user types. 

We now consider $K$ types of buyers indexed by $i=1, 2, \cdots, K$, each defined by the triple $(q_i, \epsilon_i, b_i)$ with an associated acceptance region $T_i$. 
We will use the notation 
\begin{eqnarray*}
max_i=(x_i^*,p_i^*) = \mbox{argmax}_{(x, p)\in T_i} U(x, p)
\end{eqnarray*} 
to denote the optimal contract if type $i$ were the only type existing. Similarly, we will use $C_i(x,p)$ to denote the cost to a type-$i$ buyer for accepting contract $(x, p)$.
%and $(x_i,p_i)$ will represent an arbitrary contract designed for a buyer of type $i$. 
%
%The seller does not know the exact type of the buyer but we assume he knows what types are out there and their distribution; 

A buyer is of type $i$ with probability $r_i$. We assume that the seller knows only this distribution of  types but not the actual type of a given buyer.  Consequently it has to %guess the buyer's type and 
design the contracts in a way that maximizes its expected payoff.  Since the payoff is measured in expectation, it turns out that it does not matter whether the seller is faced with a single buyer or multiple buyers as long as they are drawn from the same, known type distribution and the seller has sufficient bandwidth to honor its contracts. %The only difference is that in practice if there are multiple buyers the seller needs to be prepared for the possibility of entering into and honoring multiple contracts. We emphasize, however, that the contract design process remain exactly the same.  
For this reason throughout our discussion we will take the view of a single buyer drawn from a certain type distribution. In Section \ref{sec:discussion} we discuss the case when the seller has limited bandwidth to trade.  

%For each contract designed, the seller knows exactly which contract each possible type would pick. Thus, it is beneficial to design 
Consider a set of contracts 
%\begin{eqnarray*}
$\mathbb{C}=\{(x_1,p_1),...,(x_K,p_K)\}$
%\end{eqnarray*} 
designed by the seller with the intention that a buyer of type $i$ prefers $(x_i,p_i)$.  %.  %A buyer of type $i$ will select $(x_i,p_i)$ 
This is true iff %and only if the following set of equations are satisfied: %(let $C_i(x,p)$ denote the cost for buyer $i$)
\begin{eqnarray*}
C_i(x_i,p_i)\leq C_i(x_j,p_j) &  \forall j\neq i ~. 
\end{eqnarray*}
%In other words, the contract designed for one type of buyer must be as good as any other contract from this buyer's point of view. 
Let $R_i(\mathbb{C})$ denote the contract that a type-$i$ buyer selects given a set %of contracts 
$\mathbb{C}$. Then $R_i(\mathbb{C})=\text{argmin}_{(x,p)\in\mathbb{C}}~C_i(x,p)$
%\begin{eqnarray*}
%R_i(\mathbb{C})=\underset{(x,p)\in\mathbb{C}}{\text{argmin}} & C_i(x,p)
%\end{eqnarray*}
and the seller's expected utility for a given $\mathbb{C}$ is 
%\begin{eqnarray*}
$E[U(\mathbb{C})]=\sum_iU(R_i(\mathbb{C}))r_i$
%\end{eqnarray*}
%where $r_i$ is the \emph{a priori} probability that the buyer is of type $i$. 

In this section we consider the case where different types share the same channel condition $b_i=b, i=1, \cdots, K$, which is also known to the seller.  As mentioned earlier, this models the case where the condition is primarily determined by the seller's primary user traffic.  An example of the acceptance regions of three buyer types are shown in Figure \ref{fig:sameb}; note that $max_i$'s need not be ordered in $i$. %Under the common channel condition case, we denote the by assuming that the channel condition is known both to the seller and the buyer. It can be viewed as the channel condition b is measured by the owner, then, informed to the buyer. 
%Under this assumption, 
There are two possible cases: (1) the seller can announce as many contracts as he likes ($M=K$; note that there is no point in designing more contracts than \rev{the number of types}); (2) the seller is limited to at most $M<K$ contracts.  Below we fully characterize the optimal contract set in both cases.
\begin{figure}[h!]
\vspace{-5pt}
  \centering
    \includegraphics[width=0.4\textwidth]{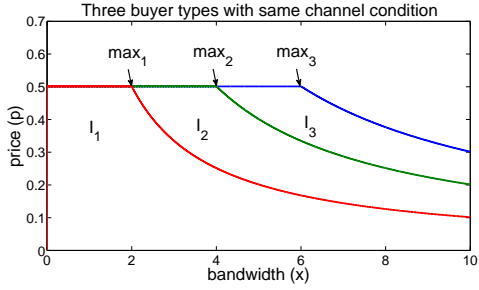}	
  \caption{Three buyer types with common $b$}
\label{fig:sameb}
\vspace{-20pt}
\end{figure}

%\com{consider how much to explain the following results...} 

\begin{theorem}
When $M=K$, the contract set that maximizes the seller's profit is $(max_1, max_2, ..., max_K)$.
\end{theorem}
This result holds for the following reason. As shown in Figure \ref{fig:sameb}, with a constant $b$, the intersection points of all acceptance regions are on the same line $p=b$.  For a buyer of type $i$, all points to the left of $max_i$ on this line cost the same as $max_i$, and all points to its right are outside the buyer's acceptance region.  Therefore the type-$i$ buyer will select the contract $max_i$ given this contract set (see earlier discussion on how the seller can always incentivize this contract over others with equal cost). Since this is the best the seller can do with a type-$i$ buyer (see Theorem \ref{thm:single_buyer}) this set is optimal for the seller. It is also relatively straightforward to obtain a similar results in the case of $M<K$ given next.

%Given the contract set $(max_1,...,max_K)$, a type-$i$ buyer will pick $max_i$ because every contract with $p=b$ and $x\leq x_i^*$ will be equal cost and every contract with $p=b$ and $x>x_i^*$ will be outside the acceptable region. Fig. \ref{fig:sameb} shows an example of the acceptable region and $max_i$s of each type when the channel condition is common.

\begin{lemma}
When $M<K$ and $\forall b_i=b$, the optimal contract set is a subset of $(max_1,..., max_K)$.\label{lemma:m<k}
\end{lemma}
\begin{proof}
Assume the optimal contract $\mathbb{C}$ is not a subset of $(max_1,..., max_K)$. Then it must consists of some contract points from at least one of the $I_i$ regions as demonstrated in Figure \ref{fig:sameb}. Let these contracts be $A_i\subset I_i$ and $\bigcup_iA_i=\mathbb{C}$. For each non-empty $A_i$, we replace it by the contract $max_i$ and call this new contract set $\mathbb{C'}$. The proof is to show that this contract set generates profit at least as large as the original one. For each type-$i$ buyer that picked some contract $(x,p)\in A_j$ from the optimal contract $\mathbb{C}$, it must had a type greater than or equal to $j$ otherwise $(x,p)$ is not in its acceptance region. In the contract set $\mathbb{C'}$, type-$i$ will now pick $max_j$ or $max_l$ with $l>j$. The choice of each possible type of buyer picks from $\mathbb{C'}$ is at least as profitable as the one they picked from $\mathbb{C}$. Thus, the expected profit of $\mathbb{C'}$ is at least as good as $\mathbb{C}$.
\end{proof}

This lemma suggests the following iterative way of finding the optimal contract set without \rev{having} to solve what would seem like a combinatorial problem. 
%\begin{definition}
Define function $g(m,i)$ as the the maximum expected profit for the seller by picking contract $max_i$ and selecting optimally $m-1$ contracts from the set $(max_{i+1},...,max_K)$. 
%\end{definition}
Note that if we include $max_i$ and $max_j$ ($i<j$) in the contract set but nothing else in between $i$ and $j$, then a buyer of type $l$ ($i\leq l<j$) will pick contract $max_i$.  These types contribute to an expected profit of $x_i^*(b-c) \sum_{l=i}^{j-1}r_l$. At the same time, no types below $i$ will select $max_i$ (as it is outside their acceptance regions), and no types at or above $j$ will select $max_i$ (as for them $max_j$ is preferable).

The function $g(m,i)$ can be recursively obtained as follows: 
\begin{eqnarray*}
g(m,i)=\underset{j: i<j\leq K-m+2}{max} g(m-1,j)+x_i^*(b-c)\sum_{l=i}^{j-1}r_l , 
\end{eqnarray*}
with the boundary condition $g(1,i)=x_i^*(b-c)\sum_{l=i}^{K}r_l$.

Finally, it should be clear that the maximum expected profit for the seller is given by $\max_{1\leq i\leq K} g(M,i)$, and the optimal contract set can be determined by going backwards: first determine $i^*_M=\arg\max_{1\leq i\leq K} g(M, i)$, then $i^*_{M-1} = \arg\max_{1\leq i\leq K-1} g(M-1, i)$, and so on.

\begin{theorem}
The set $\{max_{i^*_1}, max_{i^*_2}, \cdots, max_{i^*_M}\}$ obtained using the above procedure is optimal and its expected profit is given by $g(M, i^*_M)$. 
\end{theorem}
%\vspace{-10pt}

\section{Private channel condition} \label{sec:private} 

We now consider multiple buyer types each with a different channel condition $b_i$, $i=1, \cdots, K$.  We will start with the special case of $K=2$ and characterize the optimal contracts in this case.  Using these results we then construct an algorithm to compute a set of contracts for the case of $K\geq 2$. % and show under what conditions it is optimal. Due to space limitation not all proofs are given; we instead use intuitive arguments to illustrate key ideas.
%\vspace{-10pt}
% two buyer case %

\subsection{Two buyer types: $K=2$}
Consider two buyer types $(q_i, \epsilon_i, b_i)$, $i=1, 2$, with probability $r_i$, $r_1 + r_2 = 1$.
We first consider the case that the seller is limited to one contract: $M=1$. 
%We first consider the case when there are only two possible types of buyer $b_i, i\in \{1,2\}$, with 

\begin{figure}
\vspace{-20pt}
\centering
{\includegraphics[width=0.23\textwidth,height=0.2\textwidth]{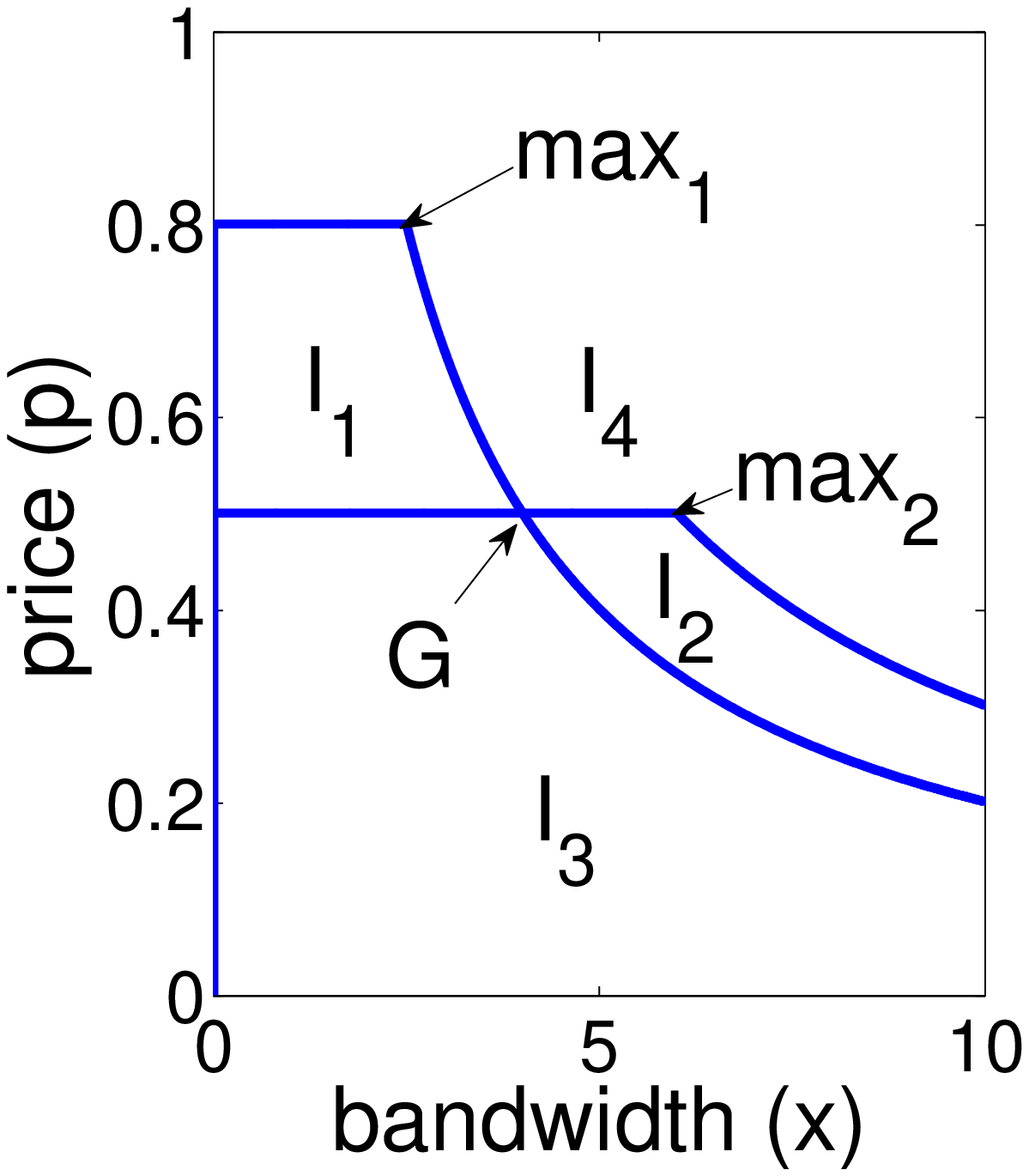}}
{\includegraphics[width=0.23\textwidth,height=0.2\textwidth]{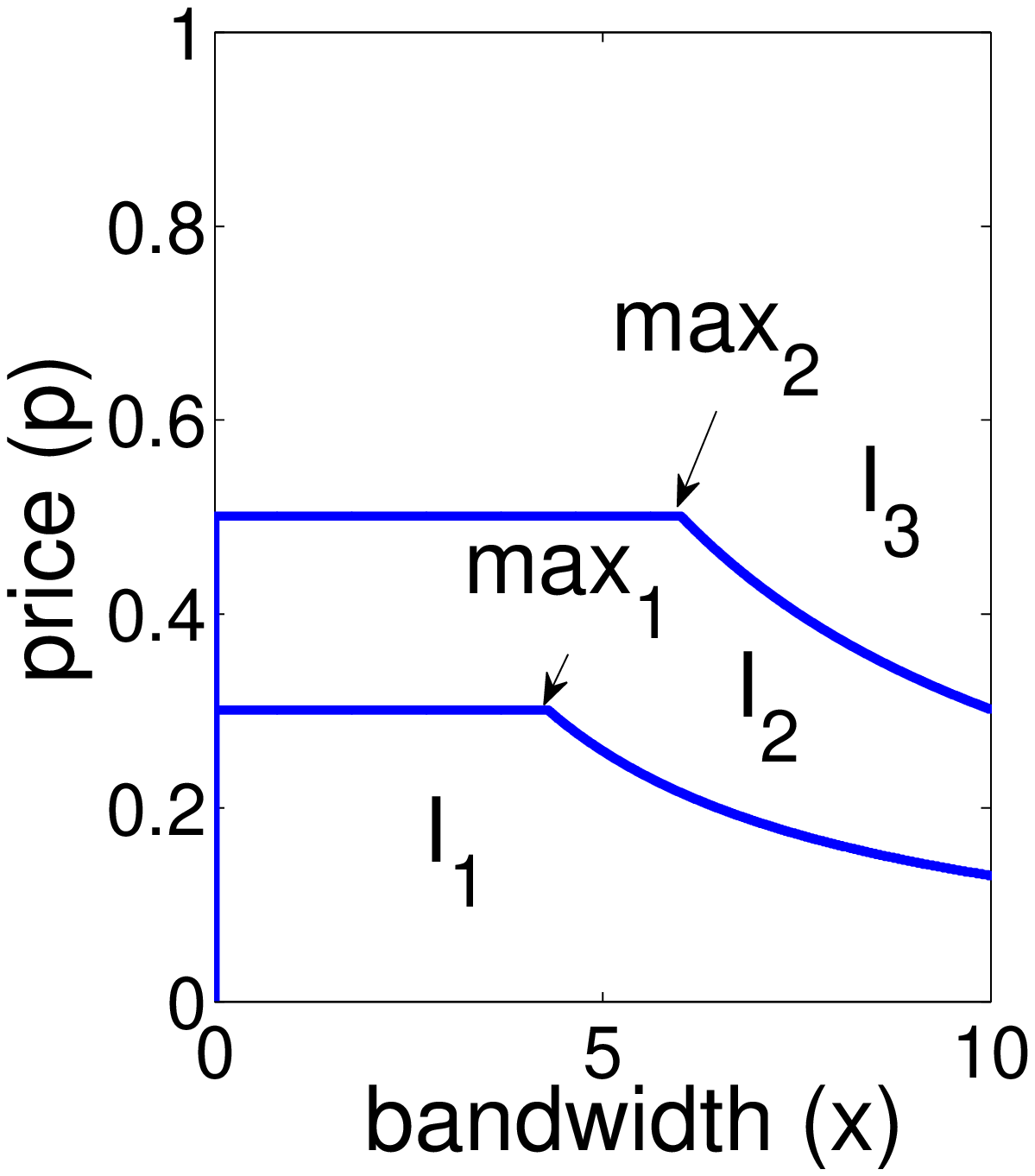}}
\caption{(left) $max_1\notin T_2$ and $max_2\notin T_1$;  (right) $max_1\in T_2$} % or $max_2\in T_1$}
\label{fig:nooverlap}
\vspace{-20pt}
\end{figure}
%
%\begin{figure}[ht]
%\centering
%\subfigure[Caption of subfigure 1]{
%    \rule{4cm}{3cm}
%    \label{fig:subfig1}
%}
%\subfigure[Caption of subfigure 2]{
%    \rule{4cm}{3cm}
%    \label{fig:subfig2}
%}
%\caption[Optional caption for list of figures]{Caption of subfigures}
%\label{fig:subfigureExample}
%\end{figure}

%\subsubsection{$M=1$ The seller can only determine one contract.}

\begin{theorem}
The optimal contract when $K=2$ and $M=1$ is as follows: 
\begin{enumerate}
\item If $max_1\notin T_2$ and $max_2\notin T_1$, 
\begin{eqnarray*}
\mbox{optimal}=\left \{
\begin{tabular}{cl}
$max_1$ & \mbox{if} ~ $r_1U(max_1)\geq r_2U(max_2)$\\
& ~~~ \mbox{and} ~ $r_1U(max_1)\geq U(G)$\\
$max_2$ & if ~ $r_2U(max_2)\geq r_1U(max_1)$\\
& ~~~ \mbox{and} ~ $r_2U(max_2)\geq U(G)$\\
$G$ & if ~ $U(G)\geq r_2U(max_2)$\\
& ~~~  \mbox{and} ~ $U(G)\geq r_1U(max_1)$\\
\end{tabular}\right .
\end{eqnarray*}
where $G$ denotes the intersecting point between acceptance region boundaries of the two types.

\item If $max_1\in T_2$.
\begin{eqnarray*}
optimal=\left \{
\begin{tabular}{cc}
$max_1$ & if ~ $U(max_1)\geq r_2U(max_2)$\\
$max_2$ & if ~ $r_2U(max_2)\geq U(max_1)$\\
\end{tabular}\right .
\end{eqnarray*}

\item If $max_2\in T_1$.
\begin{eqnarray*}
optimal=\left \{
\begin{tabular}{cc}
$max_2$ & if $U(max_2)\geq r_1U(max_1)$\\
$max_1$ & if $r_1U(max_1)\geq U(max_2)$\\
\end{tabular}\right .
\end{eqnarray*}

\end{enumerate}

\end{theorem}
The above result is illustrated in Figure \ref{fig:nooverlap} and can be argued by showing the profit of every contract in a particular region (such as $I_1$) is no greater than some specific contract (such as $max_1$). Take the case $max_1\notin T_2$ and $max_2\notin T_1$ for example, any point in $I_3$ is suboptimal to point $G$ %(The intersection point of the accepting region boundaries of the two types) 
because any contract in $I_3$ is acceptable by both types of buyers, but $G$ has a strictly higher profit than any other point in $I_3$.

%\subsubsection{$M=2$, $max_1\notin T_2$ and $max_2\notin T_1$ The seller can hand out two contracts for the buyer to choose from.} 

We now consider the case $M=2$. We shall see that providing multiple contracts can help the obtain higher profits.

\begin{theorem}
%The set  is the optimal set of contracts when 
In the case of $M=2$, $max_1\notin T_2$ and $max_2\notin T_1$, the optimal contract set is 
$\{max_1, max_2\}$. 
\end{theorem}

\begin{proof}
The set $\mathbb{C}=\{max_1, max_2\}$ gives an expected payoff of 
\begin{eqnarray*}
E[U(\mathbb{C})]&=& r_1U(R_1(\mathbb{C}))+r_2U(R_2(\mathbb{C})))\\&=&r_1U(R_1(max_1))+r_2U(R_2(max_2)). 
\end{eqnarray*}
The second equality holds because $max_1\notin T_2$ and $max_2$ $\notin T_1$ and thus type $i$ will pick $max_i$. %both types choose the $max_i$ intended for them. 
Suppose $\mathbb{C}$ is not the optimal set of 2 contracts, then there must exists some $\mathbb{C'}=\{(x_1,p_1), (x_2,p_2)\}$ such that 
\begin{eqnarray*}
E[U((\mathbb{C'}))]&=&r_1U(R_1(x_1,p_1))+r_2U(R_2(x_2,p_2)) \\
&>& E[U(\mathbb{C})]\\
&=& r_1U(R_1(max_1))+r_2U(R_2(max_2)) 
\end{eqnarray*} 
This implies either $U(R_1(x_1,p_1))>U(R_1(max_1))$, or $U(R_2(x_2,p_2))> U(R_2(max_2))$, or both, all of which contradict the definition of $max_i$.
Thus, $\{max_1, max_2\}$ is the optimal contract set.
\end{proof}

The proof as well as the intuition behind the above result are straightforward.  The next case, 
$M=2$, $max_1\in T_2$ or $max_2\in T_1$, is more complicated. % The seller can hand out at most two contracts.
%
%In this case the seller cannot hand out the same contract $\mathbb{C}=\{max_1,max_2\}$ and claim optimality because the type i buyer might not pick the contract $max_i$. 
Without loss of generality, we will assume that the type-$1$ buyer has a smaller $b_1$ ($b_1\leq b_2$), thus $max_1\in T_2$. 
%\com{Let's switch the labels in Figs 5 and 6 to match this case...} 
We first determine the optimal contract when $x^*_1\leq x^*_2$; this result is then used for the case when $x_1^*>x_2^*$.  Without loss of optimality we consider only contract pairs $\{(x_1,p_1),(x_2,p_2)\}$ where type-$i$ buyer picks $(x_i,p_i)$ instead of the other one. %It is quite simple to show that we do not lose optimality by restricting to this type of contract sets.

To find the optimal contract, we 1) first show that for each $(x_1,p_1)$ we can express the optimal $(x_2,p_2)$ in terms of $x_1$ and $p_1$; 2) then we show that $(x_1,p_1)$ must be on the boundary of $T_1$ with $x_1\leq x_1^*$;  3) using 1) and 2) we optimize the expected profit over possible choices of $x_1$. 

\begin{lemma}
When $K=2$, if $max_1 \in T_2$ and $x_1^*\leq x_2^*$, then given a contract for type-1 $(x_1,p_1)$, the optimal contract for type-$2$ must be $(x_2^*,P_2(x_1,p_1,x_2^*))$.\label{Lemma:x_2-p_2}
\end{lemma}

\begin{proof}
Given a contract $(x_1,p_1)$, the feasible region for the contract of type-$2$ buyer is the area below $P_2(x_1,p_1,x)$ as defined in Theorem \ref{Thm:equal_cost} (see Figure \ref{fig:equal_utility}).  Since the seller's profit is increasing in both $p$ and $x$, the contract that generates the highest profit is at $x_2=x_2^*$ and $p_2=P_2(x_1,p_1,x_2^*)$. %\qed
\end{proof}

\begin{lemma}
When $K=2$, if $max_1 \in T_2$ and $x_1^*\leq x_2^*$, an optimal contract for type-1 must be $p_1=b_1$ and $x_1\leq x_1^*$.\label{Lemma:p=b}
\end{lemma}
This lemma suggests that the optimal contract for type-1 lies on the straight line portion of its boundary; this can be shown using proof by contradiction.
Using Lemmas \ref{Lemma:x_2-p_2}, \ref{Lemma:p=b} and Theorem \ref{Thm:equal_cost}, the expected profit can be expressed as follows.
\begin{proof}
Lemma \ref{Lemma:p=b} can be proved in two steps. First we assume the optimal contract has $(x_1,p_1)\in T_1$, where we can increase $p_1$ by some positive $\delta>0$ but still have $(x_1,p_1+\delta)\in T_1$. By noticing that both $U(x,p)$ and $P(x,p,x')$ are increasing in $p$. We know that both $U(x_1,p_1+\delta)$ and $U(x_2^*,P_2(x_1,p_1+\delta,x_2^*)))$ are strictly larger than $U(x_1,p_1)$ and $U(x_2^*,P_2(x_1,p_1,x_2^*)))$. This contradicts the assumption that it was optimal before, thus, we know that the optimal contract for $(x_1,p_1)$ must be on the two lines (the upper boundary of $T_1$) defined in Theorem \ref{thm:acceptance_region}. Then we can exclude the possibility of having $(x_1,p_1)$ on the boundary of $T_1$ with $x_1>x_1^*$ by comparing the contract $(x_1^*, b_1)$ with such a contract. 
\end{proof}
\begin{figure}
\vspace{-20pt}
  \centering
    \includegraphics[width=0.4\textwidth]{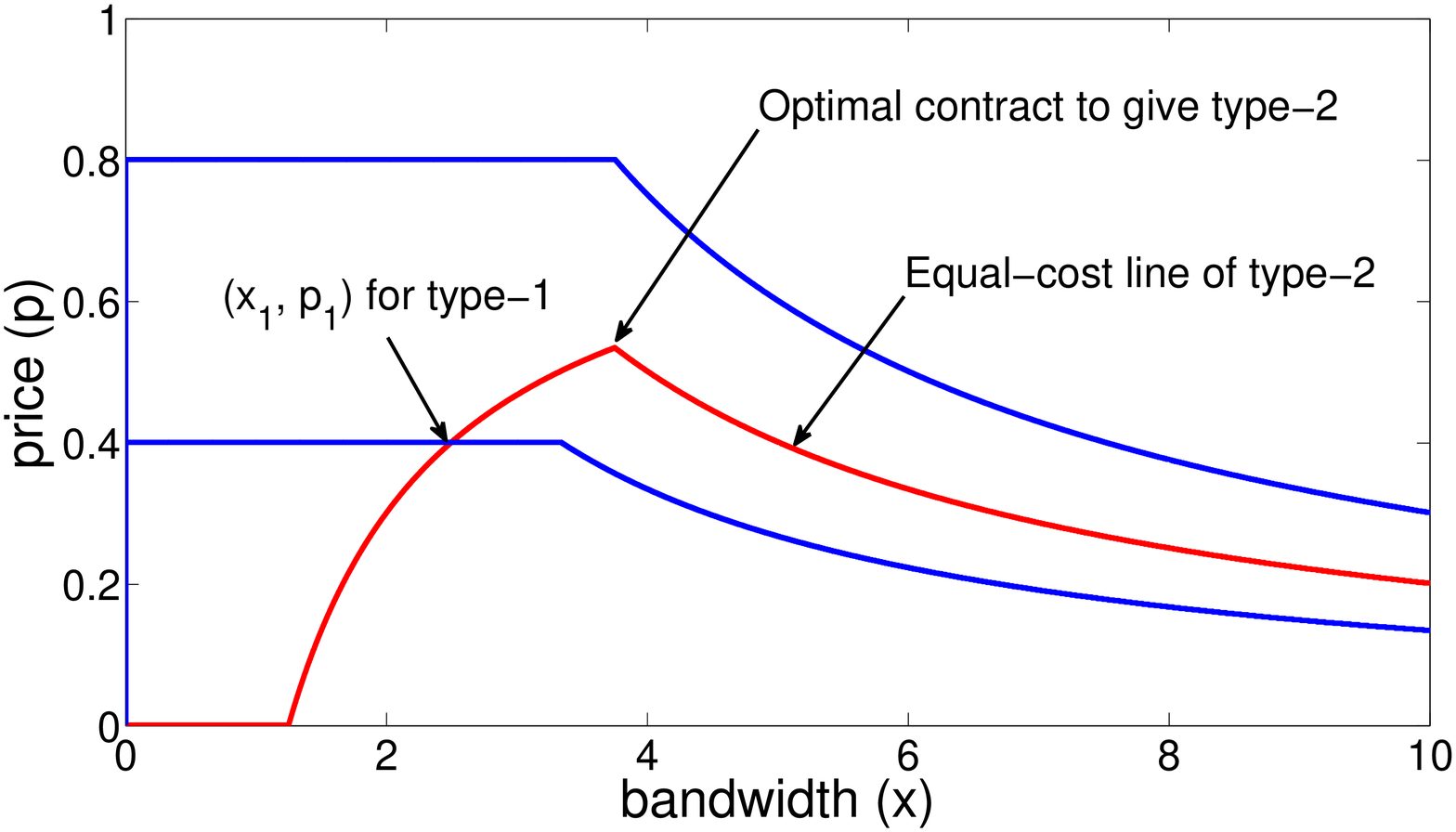}
  \caption{The regions to distinguish type-2 given $(x_1,p_1)$}
\label{fig:equal_utility}
\vspace{-20pt}
\end{figure}
\begin{eqnarray*}
E[U(\mathbb{C})] %&=&r_1U(x_1,p_1)+r_2U(x_2,p_2)\\
&=&r_1U(x_1,p_1)+r_2U(x_2,P_2(x_1,p_1,x_2^*))\\
&=&r_1U(x_1,b_1)+r_2U(x_2^*,b_2-\frac{x_1}{x_2^*}(b_2-b_1))\\
&=&r_1x_1(b_1-c)+r_2x_2^*(b_2-\frac{x_1}{x_2^*}(b_2-b_1)-c)\\
\frac{\partial E[U(\mathbb{C})]}{\partial x_1}&=&r_1(b_1-c)-r_2(b_2-b_1)
\end{eqnarray*}
The $x_1$ acheiving the optimal contract $\mathbb{C}$ is given by,
\begin{align*}
x_1&=
\begin{cases}
0 & \text{if~} r_1(b_1-c)-r_2(b_2-b_1)<0\\
x_1^* & \text{if~} r_1(b_1-c)-r_2(b_2-b_1)>0
\end{cases}\\
\mathbb{C}&=
\begin{cases}
\{max_2\} ~~~~~~~~ \text{if~} r_1(b_1-c)-r_2(b_2-b_1)<0\\
\{max_1,(x_2^*,b_2-\frac{x_1^*}{x_2^*}(b_2-b_1)) \} ~~~~~~~~~~~~~ \text{o.w.}
\end{cases}
\end{align*}

This result shows two operating regimes: 1) When $\frac{r_1}{r_2}<\frac{b_2-b_1}{b_1-c}$, type-$2$ is more profitable and the seller will distribute $max_2$. In this case there is no way to distribute another contract for type-$1$ without affecting the behavior of type-$2$. Consequently, the seller only distributes one contract. 2) When $\frac{r_1}{r_2}>\frac{b_2-b_1}{b_1-c}$,  type-$1$ is more profitable and the seller will distribute $max_1$. After choosing $max_1$, the seller can also choose $(x_2^*,b_2-\frac{x_1^*}{x_2^*}(b_2-b_1))$ for the type-$2$ buyer without affecting the type-$1$ buyer's choice. As a result, the seller distributes a pair of contracts to get the most profit.  The optimal contract for $x_1^*>x_2^*$ can be determined with a similar argument.

Again, we can prove that the optimal contract must have $p_1=b_1$ and $x_1\leq x_1^*$. The difference is that when $x_1^*>x_2^*$, the expression for $(x_2^*, P_2(x_1,p_1,x_2^*))$ has two cases depending on whether $x_1>x_2^*$ or $x_1\leq x_2^*$. 
\begin{align*}
&E[U(\mathbb{C})]
=\\
&~~\begin{cases}
r_1U(x_1,b_1)+r_2U(x_2^*,b_2-\frac{x_1}{x_2^*}(b_2-b_1)) & \text{if~} x_1\leq x_2^*\\
r_1U(x_1,b_1)+r_2U(x_2^*,\frac{x_1b_1}{x_2^*}) & \text{if~} x_1> x_2^*
\end{cases}\\
&\frac{\partial E[U(\mathbb{C})]}{\partial x_1}=
\begin{cases}
r_1(b_1-c)-r_2(b_2-b_1) & \text{if~} x_1\leq x_2^*\\
r_1(b_1-c)+r_2b_1 & \text{if~} x_1> x_2^*
\end{cases}
\end{align*}

To summarize, when $r_1(b_1-c)-r_2(b_2-b_1)>0$, $E[R(\mathbb{C})]$ is strictly increasing in $x_1$ and we know that $x_1=x_1^*$ maximizes the expected profit. When $r_1(b_1-c)-r_2(b_2-b_1)<0$, $E[R(\mathbb{C})]$ is decreasing in $x_1$ if $x_1\in[0,x_2^*]$ and increasing in $x_1$ if $x_1\in [x_2^*,x_1^*]$. We can only conclude that either $x_1=0$ or $x_1=x_1^*$ maximizes the expected profit.

\begin{align*}
&x_1=
\begin{cases}
0 \text{~or~} x_1^* & \text{if~} r_1(b_1-c)-r_2(b_2-b_1)<0\\
x_1^* & \text{if~} r_1(b_1-c)-r_2(b_2-b_1)>0
\end{cases}\\
&\mathbb{C}=\\
&\begin{cases}
max_2 \text{/} \{max_1, (x_2^*,\frac{x_1^*b_1}{x_2^*})\} & \text{if~} r_1(b_1-c)-r_2(b_2-b_1)<0\\
\{max_1, (x_2^*,\frac{x_1^*b_1}{x_2^*})\} & \text{if~} r_1(b_1-c)-r_2(b_2-b_1)>0
\end{cases}
\end{align*}

In the first condition, we can calculate the expected profit of the two contract sets and pick the one with the higher profit.

\subsection{$K$ buyer types, $K> 2$}\label{section:alg}

%In this section, we explore the general case when there can be any number of buyer types $K$. 
The previous section gives the explicit solution to the contract design problem when $K=2$. 
When $K> 2$ we no longer have explicit solutions; even numerically searching for the optimal contract set becomes very complicated.  For instance, even if we assume that both $x$ and $p$ are from discrete sets, with $X$ and $P$ possible values, respectively, the search must be done over the space of all possible sets of $K$ different contracts, on the order of $(XP)^K$. % possible sets with $K$ different contracts. %\com{shouldn't it be $(XP)$ choose $K$ such sets?} 
In general $X$ and $P$ both take on real values, making the search space uncountable.
In order to reduce the complexity we will need to exploit special properties of the problem.
%
%We first note that in general we can decompose the optimization problem into two parts given in the next lemma. %Lemma. \ref{lemma:seperate}.
%\begin{lemma}
%\begin{eqnarray*}
%\hspace{-0.4in} &\underset{\forall ((x_1,p_1),...,(x_K,p_K))}{max} & U((x_1,p_1),...,(x_K,p_K)) \\
%\hspace{-0.4in} &\qquad \qquad = \underset{\forall (x_1,...,x_K)}{max}~~\underset{\forall (p_1,...,p_K)}{max}
%& U((x_1,p_1),...,(x_K,p_K))
%\end{eqnarray*}\label{lemma:separate}
%\end{lemma}
%
%Consider first the inner maximization
%\begin{eqnarray}
%\underset{\forall (p_1,...,p_K)}{max}U((x_1,p_1),...,(x_K,p_K)), \label{eqn:inner}
%\end{eqnarray}
%which finds the set $(p_1,...,p_K)$ that maximizes the profit for some given set $(x_1,...,x_K)$.
%This problem can be solved using the following procedure. 
%The idea here will be to
%
We first reindex the buyer types such that $b_1\leq ...\leq b_K$, under certain conditions we will determine a procedure which finds the optimal contract.
% and sequentially determine the contracts for the $i$-th type, in the order $i=1, \cdots, K$. Specifically, at step $i+1$ we select the contract $(x_{i+1}, p_{i+1})$ for the type-($i+1$) buyer from the equal-cost line of the previous contract $(x_i,p_i)$, i.e., $p_{i+1} = P_{i+1}(x_i, p_i, x_{i+1})$, so that the type-($i+1$) buyer prefers $(x_{i+1},p_{i+1})$ to $(x_i,p_i)$.  It turns out this procedure determines a set of contracts that maximize Eqn. (\ref{eqn:inner}) under a {\em monotonicity condition} defined next.

%\com{For the case with multiple spectrum buyers with different types, the authors provide a solution of the optimization problem assuming the monotonic condition; if the probability that the channel state is good is small, then the spectrum buyer tends to lease less amount of spectrum. The authors, however, do not justify this condition. This could be a reasonable assumption, but in some cases, such spectrum buyers might want to lease enough amount of spectrum at a low cost to compensate the bad channel state. Discussing this issue more in detail in Section 5 could improve the quality of the paper. }

\begin{definition}
The buyer types are said to satisfy a monotonicity condition (MC), if $\forall i,j$, $b_i\leq b_j$ implies $x_i^*\leq x_j^*$. (Remember we will reindexed the types so that $b_1\leq ...\leq b_K$ and thus $x_1^*\leq ... \leq x_K^*$)
\end{definition}

This monotonicity condition (MC) says that the amount a buyer willing to buy is strictly increasing in the quality it gets from buying the secondary spectrum. This condition leads to  special properties which allows us to construct simpler ways to find the optimal contracts. %We start by first stating two Theorems.

\begin{theorem}
When the MC is satisfied, $b_i \leq b_j$ and $x<x'$ implies $P_i(x',p',x)\geq P_j(x',p',x)$.\label{thm:x<x'}
\end{theorem}

\begin{theorem}
When the MC is satisfied, $b_i \leq b_j$ and $x>x'$ implies $P_i(x',p',x)\leq P_j(x',p',x)$.\label{thm:x>x'}
\end{theorem}

\begin{proof}
The proofs for Theorem \ref{thm:x<x'} and Theorem \ref{thm:x>x'} are moved to the Appendix.
\end{proof}

\begin{lemma}
When the MC is satisfied, the optimal contract such that type i buyer picks $(x_i,p_i)$ for all i must have $x_1\leq ... \leq x_K$.
\end{lemma}

\begin{proof}
Let $(x_i, p_i)$ denote the contract designed for the type i buyer. Consider now the contract for the type j buyer where $b_j<b_i$ and $x_j>x_i$. From Theorem \ref{thm:x>x'} we know that $P_j(x_i,p_i,x_j)\leq P_i(x_i,p_i,x_j)$ when the MC is satisfied. This implies that whatever $p_j$ we determined, if the type j buyer prefers $(x_j,p_j)$ over $(x_i,p_i)$ then the type i buyer must think the same way. From the IC constraint, the type j buyer has to prefer the $(x_j,p_j)$ over $(x_i,p_i)$. Thus, we must have $x_j\leq x_i$ in the optimal contract where each type of buyer selects its own designated contract.
\end{proof}

\begin{figure}
\vspace{-20pt}
  \centering
    \includegraphics[width=0.4\textwidth]{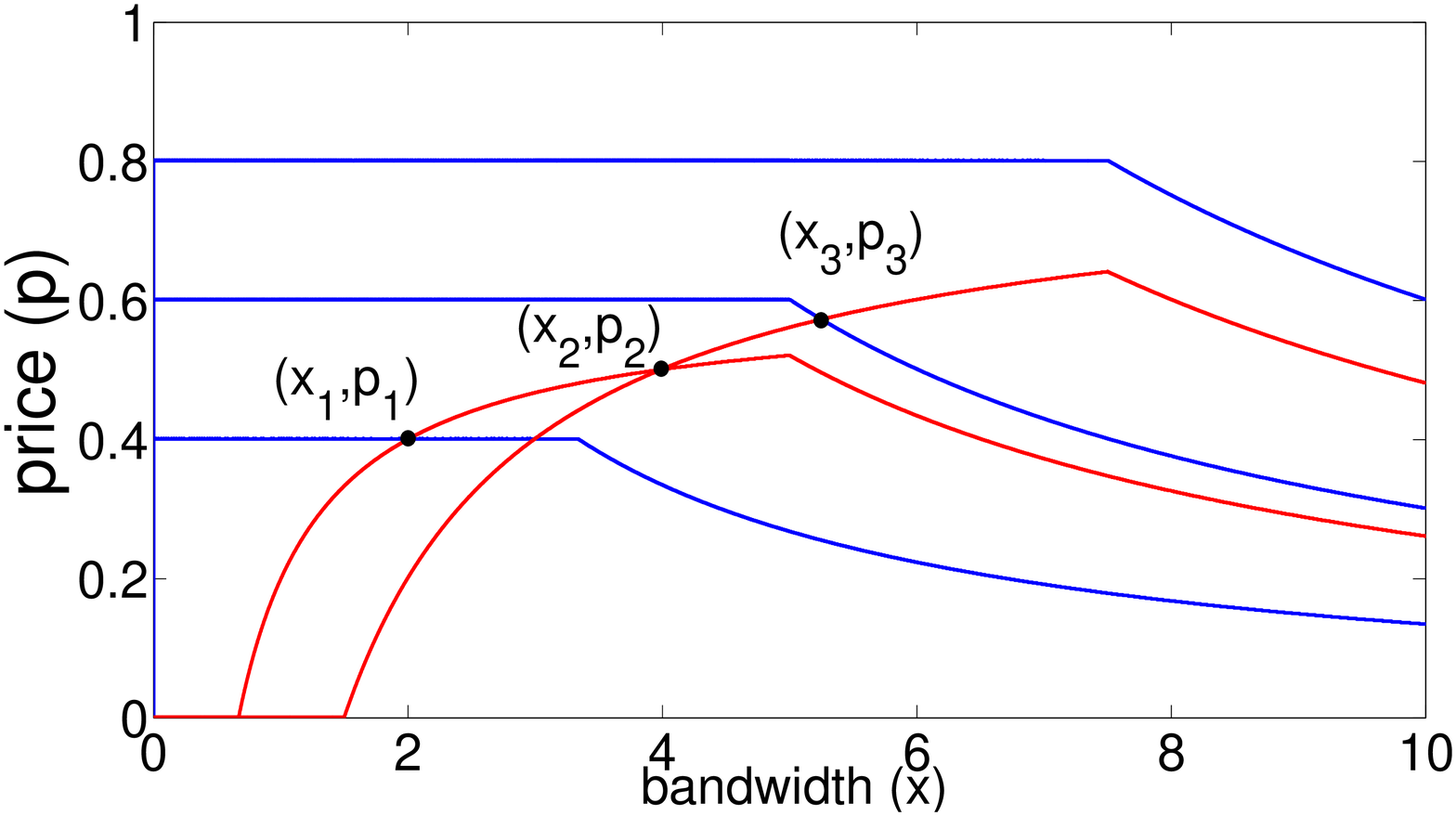}	
  \caption{Example of a possible optimal contract}
\label{fig:nested}
\vspace{-20pt}
\end{figure}
%By Theorem \ref{thm:optimal} we can express the optimization problem for the seller as,
%\begin{align*}
%& \underset{\forall (x_i,p_i)}{max}U((x_1,p_1),...,(x_K,p_K))\\
%&=\underset{0\leq x_1 ...\leq x_K}{max}U((x_1,P_1(0,0,x_1)),...\\
%&\qquad\qquad\qquad\qquad\qquad ,(x_K,P_K(x_{K-1},P_{K-1}(\cdot),x_K)))~~~~~~~~~~~~~~~~~~~
%\end{align*}
%
%Now the maximization problem is purely defined on the space of $x_i$ with the $p_i$ defined recursively on the values of $(x_{i-1},p_{i-1})$. We notice that the function $P_i$ is seperated into two parts, which causes a problem if we want to solve the maximization by taking the derivative. The Lemma below shows that we can further restrict the region for $x_i$ when searching for the optimal contract.
\begin{lemma}\label{lem:set_x_2}
When the MC is satisfied, the optimal contract must have $x_i\leq x_i^*$ $\forall i=1...K$.\label{Lemma:x_i}
\end{lemma}
\begin{proof}
Proof by contradiction. Consider some optimal contract having $x_i>x_i^*$, we show that replacing $x_i=x_i^*$ is actually better. By Theorem \ref{thm:optimal}, we know that $p_i=P_i(x_{i-1},p_{i-1},x_i)$ and by definition of $P_i$ it is better off to the seller by providing $x_i^*$ instead if we only consider the profit from the type $i$ buyer. Now, by Theorem \ref{thm:x<x'} $P_{i+1}(x_i,p_i,x_i^*)\leq P_i(x_i,p_i,x_i^*)$. Also, because $P_i(x',p',x)$ is a strictly increasing function in $p'$. The price $p_{i+1}$ is strictly higher for assigning $x_i^*$ instead of $x_i$. This results in every $p_j$ $j>i$ is strictly increased and the payoff change must be positive. The only question is whether we can assign $x_i^*$ without affecting the contracts $(x_j,p_j)$ $j<i$. The answer is if $\forall j<i$ $x_j\leq x_j^*$ we can do it. By mathematical induction, we can again prove that for all $i=1...K$ $x_i\leq x_i^*$. An example is illustrated in Figure \ref{fig:nested}.
\end{proof}

This result allows us to restrict our search for the optimal contract to the set where $x_i\leq x_i^*$. We can further simplify our search by expressing the values $p_i$, $\forall i=1...K$ as functions of $x_i$ $\forall i=1...K$, by the following theorem. 

\begin{theorem}
Given a set $x_1\leq ...\leq x_K$, define $(x_0,p_0)=(0,0)$ and find the contracts $(x_i,p_i)=(x_i, P_{i}(x_{i-1},p_{i-1}, x_{i}))$ in the order $i=1...K$. When the MC is satisfied this procedure produces a contract set that maximizes the seller's profit, where each type-$i$ buyer accepts $(x_i,p_i)$. \label{thm:optimal}
\end{theorem}

\begin{proof}
\begin{enumerate}
\item Each buyer of type $i$ picks $(x_i,p_i)$.\\
Induction hypothesis: At each step, when we pick contract $(x_i,p_i)$ $\forall i=0...K$, each buyer type-$j$ with $j<i$ prefers contract $(x_j,p_j)$ and each buyer type-$j$ with $j\geq i$ prefers contract $(x_i,p_i)$. 
\begin{enumerate}
\item Base Case: When picking $(x_0,p_0)=(0,0)$, it is clear that each buyer type is greater than $0$ and each buyer prefers the only contract that is the same as (as good as) not buying.
\item Assume the induction hypothesis is true when picking $(x_i,p_i)$, we will show that the hypothesis is also true for $(x_{i+1},p_{i+1})$.
Assume the hypothesis is true for step $i$ means we have determined the contracts $((x_1,p_1),...,(x_i,p_i))$ and a type-$j$ buyer ($j\leq i$) prefers $(x_j,p_j)$ over other contracts, while a type-$j$ buyer ($j>i$) prefers the $i$th contract over all contracts.
\begin{eqnarray*}
p_{i+1}=(x_{i+1},P_{i+1}(x_{i},p_{i}, x_{i+1}))
\end{eqnarray*}
 By Theorem. \ref{thm:x>x'} and $x_{i+1}>x_{i}$, $\forall j\leq i$
\begin{eqnarray*}
p_{i+1}=P_{i+1}(x_{i},p_{i}, x_{i+1})\geq P_{j}(x_{i},p_{i}, x_{i+1})
\end{eqnarray*}
The contract $(x_{i+1},p_{i+1})$ is above the equal cost line of the contract $(x_i,p_i)$ for buyer type less than or equal to $i$. Which means they prefer the $i$th contract over the $i+1$th contract. But from step $i$, they prefer their own contract over existing contracts. Thus, buyer $j$ $(j\leq i)$ prefers $(x_j,p_j)$ over all contracts.
By Theorem. \ref{thm:x<x'} and $x_{i+1}>x_{i}$, $\forall j\geq i+1$
\begin{eqnarray*}
p_{i+1}= P_{i+1}(x_i,p_i,x_{i+1})\leq P_j(x_i,p_i,x_{i+1})
\end{eqnarray*}
Thus, the contract $(x_{i+1},p_{i+1})$ is below the equal cost line of the contract $(x_i,p_i)$ for buyer type $j>i$. Which means they prefer $(x_{i+1},p_{i+1})$ over $(x_i,p_i)$. But from step $i$, they prefer the $(x_i,p_i)$ contract over all existing contracts. This shows that the hypothesis is true for step $i+1$.
\item By Mathematical Induction, the hypothesis is true for all $i\leq K$.
\end{enumerate}
\item This process results in the highest profit.\\
Since the $x_i's$ are fixed, the only way one could increase the buyer's profit is to increase one of the $p_i$'s. We will show that this is not possible. Assume there exists some contract with the contract set $(x_1,p_1')...(x_K,p_K')$ with some $p_i'>p_i$, by the increasing property of $P_i$ (Lemma \ref{Lemma:strict_increasing}) we need $p_{i-1}'>p_{i-1}$ to insure that type-$i$ buyer picks $(x_i,p_i')$. By induction, we can show that it must be that $(p_1'>p_1)$. Since $p_1=P_1(0,0,x_1)=b_1$, $(x_1,p_1)$ is already on the boundary of acceptance region of the type-$1$ buyer. Thus, any contract with some $p_i'>p_i$ is not a contract where each buyer accepts its own designated contract.
\end{enumerate}
\end{proof}

%We sketch the two parts of the proof of Theorem \ref{thm:optimal}. One part is to show that each buyer will pick the contract designated for he/she in the resulting $K$ contracts; this is proved by induction. The second part is to show that this is the highest profit the seller can get, which is shown by proof by contradiction. 

Figure \ref{fig:nested} shows an example of applying this theorem with three buyer types: given $x_1=2$, $x_2=4$, $x_3=6$, $p_i$ is sequentially determined on the equal-cost line of the previous contract.  With Lemma \ref{lem:set_x_2}, the equal cost line can be restricted to the form $P_{i}(x_{i-1},p_{i-1},x_i)=b_i-\frac{x_{i-1}}{x_i}(b_i-p_{i-1})$.  %This simplifies the search for the optimal set of $x_i \leq x_i^*$. 
The expected profit of the seller can now be expressed as: 
\begin{align*}
&E[R(\mathbb{C})] 
 =\underset{x_1,..,x_K}{max} r_1x_1(b_1-c)\\
&~~~~~~~~~~~~~~~~~+...+r_ix_i(p_i-c)+...+r_Kx_K(p_K-c)\\
&=\underset{x_1\leq ... \leq x_K}{max} r_1x_1(b_1-c)+...+r_ix_i(P_{i}(x_{i-1},p_{i-1},x_i)-c)\\
&~~~~~~~~~~~~~~~~~+...+r_Kx_K(P_{K}(x_{K-1},p_{K-1},x_K)-c)
\end{align*}
By plugging in the values of $p_i=P_{i}(x_{i-1},p_{i-1},x_i)=b_i-\frac{x_{i-1}}{x_i}(b_i-p_{i-1})$ recursively. Each term in the optimization problem can be simplified to 
\begin{eqnarray*}
r_ix_i(p_i-c)=r_i(x_i(b_i-c)-\sum_{j=1}^{i-1} x_j(b_{j+1}-b_j))
\end{eqnarray*}
By simplifying and separate the terms with respect to $x_i$, the expected profit of the seller can be expressed as,
\begin{eqnarray*}
E[R(\mathbb{C})]=\underset{x_1\leq ... \leq x_K}{max} \sum_{i=1}^K x_i \left(r_i(b_i-c)-(b_{i+1}-b_i)\sum_{j=i+1}^Kr_j \right)
\end{eqnarray*}
Firstly, we observe that the above expression is linear in every $x_i$. Thus differentiating with respect to $x_i$ we get a constant: 
\begin{eqnarray*}
\frac{\partial E[R(\mathbb{C})]}{\partial x_i}=r_i(b_i-c)-\sum_{j=i+1}^K r_j(b_{i+1}-b_i)
\end{eqnarray*}
Secondly, because the term $\frac{\partial E[R(\mathbb{C})]}{\partial x_i}$ does not depend on any $x_j$, the optimizer can be easily determined. %\rev{It is either $\frac{\partial E[R(\mathbb{C})]}{\partial x_i}>0$ or $\frac{\partial E[R(\mathbb{C})]}{\partial x_i}\leq 0$. 
When $\frac{\partial E[R(\mathbb{C})]}{\partial x_i}>0$ we want to make $x_i$ as large as possible ($\leq x_i^*$); when $\frac{\partial E[R(\mathbb{C})]}{\partial x_i}<0$ we want to make $x_i$ as small as possible. % (but larger than $x_{i-1}$).} 
This leads us to the following algorithm which finds the optimal set of $(x_1,...,x_K)$. The variable $LD$ (Last Determined) below is used to keep track of the last type for which we have already determined its value.

\begin{description}
\item[Step 1]. %At this step, we want to determine the value of $x_K$.\\
Since $\frac{\partial E[R(\mathbb{C})]}{\partial x_K}=r_K(b_K-c)>0$, %we want to increase $x_K$ as much as we can. Thus, we let 
let $x_K=x_K^*$.\\
Set $LD= K$, $i=K-1$. Go to Step 2.
\item[Step 2]. %At this step, we want to determine the value of $x_{i}$.\\
%Calculate 
$\pi := (b_{i}-c)\sum_{j=i}^{LD-1}r_j-(b_{i+1}-b_i)\sum_{j=LD}^K r_j$.\\
If $i$ equals to $1$ go to Step 5; %\\
else if $\pi>0$, go to Step 3; 
else if $\pi\leq 0$, go to Step 4.
\item[Step 3]. %$\pi >0$ means that we want to increase $x_{i}$ as much as we can. 
Set $x_{j}=x_{i}^* (\forall i \leq j < LD)$. Let $LD=i$, $i := i-1$. Go to Step 2.
\item[Step 4]. %$\pi\leq 0$ means that want to decrease $x_{i}$ as much as we can. %Let $x_{i}=x_i$. % should be pushed to the same value as $x_{i-1}$. 
Let $i := i-1$. Go to Step 2.
\item[Step 5]. If $\pi>0$, let $x_{j}=x_{1}^*$ $(\forall~1 \leq j < LD)$, else let $x_{j}=0 (\forall~1 \leq j < LD)$. Terminate.
\end{description}

This algorithm works as follows: We start from determining the value of $x_K$, then we determine $x_{K-1}$ and so on all the way to $x_1$. At step $i$ we take the derivative with respect to $x_i$. If it is better to maximize it, we assign it to be $x_i^*$. If it is better to minimize it, we push the value to $x_{i-1}$ (which we have not determined). However, we have to add the probability of occurrence $r_i$ to the value $(x_{i-1})$ we pushed to so that it reflects the weight of occurrence when determining the value $x_{i-1}$. Once we determined the value for some $x_i$, every $x_j$ previously pushed to it will be assigned the same value.

Together with Theorem \ref{thm:optimal} the above algorithm produces a set of $(x_i, p_i)$'s that's optimal under the monotonicity condition.  This algorithm takes exactly $K$ steps to find the optimal contract set. While calculating the $\sum r_i$ might also take $K$ steps, with careful calculation the method can still be completed in $O(K)$ time.  By comparison, an exhaustive search method will take $O((XP)^K)$ time to find the optimal contract even if we discretize the search space of $x$ and $p$ with $X$ and $P$ possible values. When $x$ and $p$ are continuos, an exhaustive search might not even be possible.

\section{Discussion}\label{sec:discussion} 

\paragraph{A seller with limited resource} %\label{limited}

%\com{Third, the authors implicitly take some impractical assumptions. For example, the spectrum seller is assumed to have infinite amounts of excess bandwidth and can accommodate any buyer's request. The reviewer is curious to know what will happen if the bandwidth resources are limited.}

Our analysis so far has been based on the assumption that the seller has sufficient bandwidth to fulfill all accepted contracts.  We now discuss what happens when the seller's resources are limited.  In the full information case when the seller knows the type of each of a group of potential buyers, it will extract the most by offering $(b_i, x^*_i)$ to a buyer of type $i$.  Under a resource constraint, it will select a set of contracts $S$ to maximize $\sum_{i\in S} b_i x^*_i$ s.t. $\sum_{i\in S} b_i$ does not exceed the limit.  Because the seller can offer any $0<x<x^*_i$ (when $p$ is set to $b_i$), this becomes a form of the continuous (fractional) knapsack problem.
%
%the assumption that there is only one potential buyer and the seller has enough bandwidth to fullfill all contracts. Let us now consider the case when there are multiple buyers. To make the problem more tractable, we assume that the seller knows exactly which type each buyer is. Thus, the seller can hand out each contract specially designed for the buyer. Let the number of users of each type be $N_i$, $i\in K$ and the seller's resource constraint be $L$. (The seller can at most sell $L$ bandwidth in total)
%
%From the discussion in the single buyer type setup and the seller knows exactly the type of the buyer, the seller can extract the most by offering $(b_i, x^*_i)$ to a buyer. 
%
%The optimal strategy for the seller is to sell to the buyers with high $b_i$. For the buyers with low $b_i$, the seller does not want to sell to them and might as well give them a contract outside their acceptance region.
\rev{This problem is solvable by a greedy algorithm \cite{cormen2001introduction}. }
%where the seller sorts the buyer types according $b_i$ and only offers contracts to the types with $b_i>Threshold$ for some $Threshold\in R$. This $Threshold$ depends on the types and the resource constraint.
%The $Threshold$ value can be determined by a simple greedy algorithm.}
When the seller does not know the buyers' types, the resulting problem is a similar constrained optimization using type distribution information. 

\paragraph{Comparing to auction} %\label{sec:auction}

Auction has been used extensively for the allocation of spectrum on the traditional, wholesale market, and has been proposed for the secondary market as well, see e.g., \cite{wu2008multi,gandhi2007general,huang2006auction}.  Auction is a  mechanism aimed at extracting profit from the sale of rare goods for which potential buyers'  valuation is unknown and can be very hard to obtain.  The contract mechanism studied in this paper may be view as a form of sale by {\em posted price}.  Compared to auction, posted price is more often used in the sale of multiple (and potentially large quantity of) similar goods, the valuation of which is obtained through market research \cite{wang2008auction}.  Since the cost spent on market research can be amortized over multiple goods, posted price sale can be more efficient than auction which incurs cost in conducting each single auction \cite{zeithammer2006auctioning}.  It has been shown that under ideal conditions the two are equivalent in profit generation \cite{myerson1981optimal}.  As more and more license holders may be interested in the secondary market, we believe pricing schemes like the contracts studied in this paper offers a valid alternative to spectrum auction.

\paragraph{Learning buyer types and continuous type distribution} %\label{sec:learning}

We have assumed in our analysis that the seller knows a priori the buyer distribution and that this distribution is discrete. This knowledge can be obtained through online learning, 
%
%In some situations, it might not be that the seller in advance knows the type of the buyers. Maybe it is costly to investigate in the market. (to learn the buyer types) Instead, maybe the seller would rather learn the market by offering various contracts to potential buyers and see what happens. 
%In this section we consider this case 
where a stream of buyers arrive and the seller offer contracts designed not only to make profit (exploit) but also to learn the buyer type distribution (explore) by observing whether the contract is accepted or rejected. 
%
%it is setup to have multiple runs of contracts signed (sequential in time), where the buyer type is drawn from some stationary or varying distribution. The seller has to learn the types by offering contracts and see whether the buyer accepts or not.
%
%A setup can be as follows:
%\begin{enumerate}
%\item The buyer is of a type $(q,b,\epsilon)$ drawn from a distribution. (this distribution can be time varying)
%\item There are $T$ buyers coming at each time instance $1...T$.
%\item The seller offers a contract (or a set of contract), the response is only whether the buyer accepted or not.
%\item The goal of the seller is to receive maximum profit.
%\end{enumerate}
%
This can be cast as a multi-armed bandit problem with an independent reward process (assuming buyers are independently drawn from a distribution), and potentially a continuum of arms (each contract is an arm under this model) if the buyer types follow a continuous distribution.  Algorithms exist in the literature that produce sublinear regret (defined as the profit difference between the best single contract and the algorithm) in time \cite{auer2007improved}, and logarithmic regret in time when the number of arms is finite \cite{auer2002finite}. 

%We can view problem this as a multi-arm bandit problem where each contract is considered as an arm. Although the arm space is continuous in general, if we restrict ourselves to discrete possible contracts the problem becomes a finite multi-arm bandit problem. Under finite arms setup with $T$ buyers, even though the pull of an arm has a very different nature than the contract. Since we are considering only the profit, an acceptance of a contract can be viewed as a successful pull of an arm. Applying known results, there are algorithms that generate sublinear regrets (defined as the profit difference between the best single contract and the algorithm). Under certain condition of the buyer type, it is also possible to achieve sublinear regret when the bandit space is infinite.
\section{Numerical Evaluation}\label{sec:simulation}

%\com{The evaluation section of the paper is not enough to fully investigate the performance of the proposed framework. For example, the authors can provide more numerical results by varying the distribution of parameters, and/or unsatisfactory ratio of spectrum buyers.}

In this section %we numerically evaluate the performance of the algorithm introduced in the previous section.  Specifically, 
we compare the performance of contracts generated by the following methods: 

\begin{enumerate}
\item The optimal set of $M$ contracts (denoted OPT(M) in the figures) : %This is the best one can do by introducing at most $M$ contracts. (Note that there is discretization error here, the step size for $x$ is $0.5$ and the step size for $p$ is $0.1$) 
Finding this set is done by an exhaustive search over a set of discretized values $x$ and $p$ as an approximation of the uncountable choices (the step size for $x$ is $0.5$ and the step size for $p$ is $0.1$).   %all possible contract sets. 
As discussed earlier in Section \ref{section:alg}, the complexity increases exponentially in $M$. This restricts us to run at most $M=2$ in our evaluation.

\item The algorithm we introduced in the previous section (denoted ALG in the figures): %This is the set of contracts determined by calculating the contract disregarding whether the monotonicity condition holds. (Section \ref{section:alg}. As the results have 
As previously shown, ALG is optimal when the monotonicity condition holds. Since the complexity of this algorithm increases only linearly in $M$, $M$ can be on the order of thousands in our numerical evaluation.

\item A $K$-choose-1 method (denoted MAX in the figures): This is the method that selects the contract with the highest expected profit over the set $\{max_1, \max_2, \cdots, \max_K\}$:  
% is a single contract $max_i$ with the highest expected profit.
%\begin{eqnarray*}
%MAX:=
$\underset{max_i, i=1...K}{\mbox{maximize}} ~ E[U(max_i)]$. 
%\end{eqnarray*}
This is done by checking all $(b_i, x_i^*)$ pairs; the complexity increases linearly in $M$. % and is independent of the number of choices for $X$ and $P$.
\end{enumerate}

\rev{The experiments are run by increasing $K=1...7$.  For each $K$ value the parameters $(q_i, b_i,\epsilon_i,r_i)$ are independently and randomly generated from uniform distributions ($b_i\in [0,1]$, $q_i\in [0,10]$, $\epsilon_i\in [0,2]$ and $r_i\in [0,1]$ but normalized such that $\sum{r_i}=1$)} For each $K$ we record the average (in expected profit) over 12000 randomly generated cases; these are plotted in Figure \ref{fig:nonincreasing}.  We repeat the same but only for cases that satisfy the monotonicity condition; results are shown in Figure \ref{fig:increasing}. 
\begin{comment}
\begin{table}
\centering
\scriptsize
\begin{tabular}{|c|c|}
\hline
Paramter & Distribution\\ \hline
$b_i$ & $\mathcal{U}(0,1)$\\ \hline
$q_i$ & $\mathcal{U}(0,10)$\\ \hline
$\epsilon_i$ & $\mathcal{U}(0,2)$\\ \hline
$r_i$ & $\mathcal{U}(0,1)$ but normalize to $\sum{r_i}=1$\\ \hline
$K$ & $1,...,7$\\ \hline
\end{tabular}
\caption{Parameters used in the evaluation}
\label{table}
\end{table}
\end{comment}
\begin{figure}[h!]
  \centering
    \includegraphics[width=0.4\textwidth]{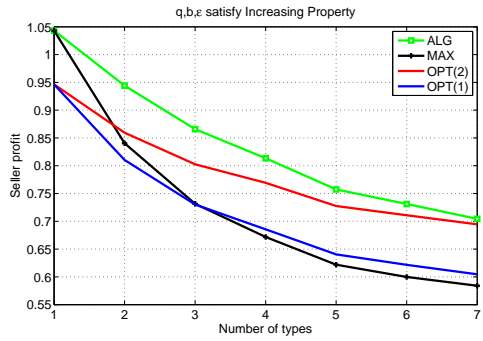}	
  \caption{Simulation results of the sellers profit versus different contracts in the general case}
\label{fig:nonincreasing}
\end{figure}

\begin{figure}[h!]
\vspace{-20pt}
  \centering
    \includegraphics[width=0.4\textwidth,height=0.25\textwidth]{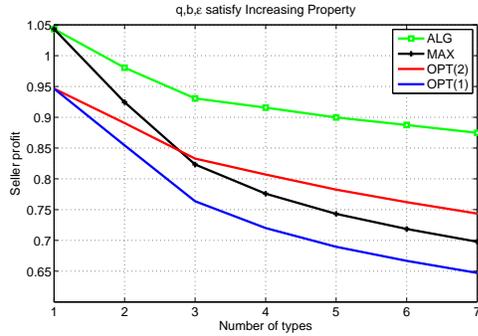}
  \caption{Simulation results of the sellers profit versus different contracts when increasing property holds}
\label{fig:increasing}
%\vspace{-20pt}
\end{figure}

Our observations are as follows. Being able to use more contracts is always better as expected (i.e., OPT(1) $\leq$ OPT(2) in all cases).  When the monotonicity condition holds, ALG is optimal and thus outperforms all other algorithms. 
When $K=1, 2$ OPT(2) should have been optimal but it falls below ALG due to the discretization error. 
%Let's only compare Optimal(2) and ALG, when $K=1,2$ Optimal(2) should perform as good as ALG. However, Optimal(2) still falls below because of the discretization error. 
When $K>2$, ALG further has the advantage of being able to use more than $2$ contracts. Recall that MAX is the optimal contract when the seller knows exactly the type; thus, MAX is optimal when $K=1$ and outperforms exhaustive search because it does not suffer from discretization error. In the general case when the monotonicity does not necessarily hold, although ALG is not always optimal it still outperforms both OPT(1) and OPT(2). Finally, when the buyer type is harder to predict (as $K$ increases), the maximum expected profit decreases.

\section{Conclusion} 

{We considered a contract design problem where a primary license holder wishes to profit from its excess spectrum capacity by selling it to potential secondary users/buyers via designing a set of profitable contracts.  We completely characterize the optimal solution in the cases where there is a single buyer type, and when multiple types of buyers share a common, known channel condition.  In the case when each type of buyers have different channel conditions we construct an algorithm that generates a set of contracts in a computationally efficient manner, and show that this set is optimal when the buyer types satisfy a monotonicity condition.} 
%We considered two cases.  Under symmetric information, we find the optimal contract that achieves maximum profit for the primary user. Under asymmetric information, we find the optimal contract if the buyer belongs to one of two types.  When there are more than two types, we present a sequential design procedure that can be shown to be optimal under certain conditions.} 

\bibliographystyle{unsrt}
\bibliography{writeup}

\appendix
Theorem. \ref{thm:x>x'}
\begin{proof}
~\\ \vspace{-15pt}
\begin{description}
\item[Case 1] $x' \leq x_i^*\leq x_j^*$\\
When $x' \leq x_i^*$ and $x' \leq x_j^*$, the equal cost lines for $x<x'$ are of the form,
\begin{eqnarray*}
P_i(x',p',x)=b_i-\frac{q_i-\epsilon_i-\delta_i}{x}\\
P_j(x',p',x)=b_j-\frac{q_j-\epsilon_j-\delta_j}{x}
\end{eqnarray*}
where we let $\delta_i=C_i(x',p')$ and $\delta_j=C_j(x',p')$. Take the derivatives with respect to $x$.
\begin{eqnarray*}
\frac{\partial P_i(x',p',x)}{\partial x}&=&(q_i-\epsilon_i-\delta_i)x^{-2}\\
\frac{\partial P_j(x',p',x)}{\partial x}&=&(q_j-\epsilon_j-\delta_j)x^{-2}
\end{eqnarray*}
By definition, $P_i(x',p',x')=p'=P_j(x',p',x')$,
\begin{eqnarray*}
p'=b_i-\frac{q_i-\epsilon_i-\delta_i}{x'}=b_j-\frac{q_j-\epsilon_j-\delta_j}{x'}
\end{eqnarray*}
Considering $b_i<b_j$, we know that $q_j-\epsilon_j-\delta_j> q_i-\epsilon_i-\delta_i$. Which implies $\frac{\partial P_j(x',p',x)}{\partial x}\geq \frac{\partial P_i(x',p',x)}{\partial x}$, and thus $P_i(x',p',x)\geq P_j(x',p',x)$, $\forall x<x'$.

\item [Case 2] $x_i^*\leq x' \leq x_j^*$\\
The equal cost lines are,
\begin{eqnarray*}
P_i(x',p',x)&=&\left\{
\begin{tabular}{ll}
$\frac{x'p'}{x}$ & $x_i^*\leq x\leq x'$\\
$b_i-\frac{q_i-\epsilon_i-\delta_i}{x}$ & $x\leq x_i^*$
\end{tabular}
\right.\\
P_j(x',p',x)&=&
\begin{tabular}{ll}
$b_j-\frac{q_j-\epsilon_j-\delta_j}{x}$ & $~~x\leq x'$
\end{tabular}
\end{eqnarray*}
Where $\delta_i=C_i(x',p')$ and $\delta_j=C_j(x',p')$. Taking the derivatives,
\begin{eqnarray*}
P'_i(x',p',x)&=&\left\{
\begin{tabular}{ll}
$-x'p'x^{-2}<0$ & $x^{i*}\leq x\leq x'$
\\
$(q_i-\epsilon_i-\delta_i)x^{-2}>0$ & $x\leq x^{i*}$
\end{tabular}
\right.\\
P'_j(x',p',x)&=&
\begin{tabular}{ll}
$(q_j-\epsilon_j-\delta_j)x^{-2}>0$ & $~~x\leq x'$
\end{tabular}
\end{eqnarray*}
This implies $P_i(x',p',x)>P_j(x',p',x)$, $\forall x$ $x_i^*\leq x\leq x'$.
\begin{eqnarray*}
P_i(x',p',x_i^*)&=&b_i-\frac{q_i-\epsilon_i-\delta_i}{x_i^*}\\
&>&P_j(x',p',x_i^*)=b_j-\frac{q_j-\epsilon_j-\delta_j}{x_i^*}
\end{eqnarray*}
Since $b_i<b_j$, we conclude that $q_j-\epsilon_j-\delta_j\geq q_i-\epsilon_i-\delta_i$. Which indicates that $\frac{\partial P_j(x',p',x)}{\partial x}\geq \frac{\partial p_i(x',p',x)}{\partial x}$ and $P_i(x',p',x)\geq P_j(x',p',x)$, $\forall x\leq x_i^*$.

\item [Case 3] $x'\geq x_j^*\geq x_i^*$\\
When $x\geq x_j^*\geq x_i^*$, the equal cost line of both types follow $x'p'=xp$. Thus, $P_i(x',p',x_j^*)=P_j(x',p',x_j^*)$. Then the case falls into Case 2 and $P_i(x_j^*,P_j(x',p',x_j^*),x)\geq P_j(x_j^*,P_j(x',p',x_j^*),x)$, $\forall x<x_j^*$.
\end{description}
\end{proof}

Theorem. \ref{thm:x<x'}
\begin{proof}
\begin{description}
\item [Case 1] $x' \leq x_i^* \leq x_j^*$\\
When $x' \leq x_i^* \leq x_j^*$, both types have equal utiliy line of the same form.
\begin{eqnarray}
P_i(x',p',x)=b_i-\frac{q_i-\epsilon_i-C_i(x',p')}{x}\nonumber \\
P_i(x',p',x)=b_j-\frac{q_j-\epsilon_j-C_j(x',p')}{x}\label{eqn:p_i}
\end{eqnarray}
By exactly the same argument as in Theorem. \ref{thm:x>x'} we can find out that. $\frac{\partial P_j(x',p',x)}{\partial x}\geq \frac{\partial p_i(x',p',x)}{\partial x}$, and thus, 
\begin{eqnarray*}
P_i(x',p',x)\leq P_j(x',p',x)~\forall x_i^*\geq x\geq x'
\end{eqnarray*}
When $x_i^*<x<x_j^*$, while $P_j(x',p',x)$ still follows the same formula (Equation. \ref{eqn:p_i}), $P_i(x',p',x)$ starts to decrease by following the line $P_i(x',p',x)=x'P_i(x',p',x_i^*)/x$. Thus, 
\begin{eqnarray*}
P_i(x',p',x)\leq P_j(x',p',x)~\forall x_i^*\leq x\leq x_j^*
\end{eqnarray*}
When $x>x_j^*$, both $i,j$ follow the form $P(x',p',x)=P(x',p',x_j^*)/x$. But $P_i(x',p',x_j^*)\leq P_j(x',p',x_j^*)$, they never cross and $P_j(x',p',x)\geq P_i(x',p',x)$ $\forall x>x_j^*$.
\item [Case 2] $x_i^*<x'<x_j^*$\\
When $x_i^*<x'<x<x_j^*$ they are of the form,
\begin{eqnarray*}
P_i(x',p',x)&=&\frac{x'p'}{x}\\
P_j(x',p',x)&=&b_j-\frac{q_j-\epsilon_j-C_j(x',p')}{x}
\end{eqnarray*}
respectively. By the same argument as in Theorem. \ref{thm:x>x'}, $P_i$ is decreasing while $P_j$ is increasing. Thus, $P_i(x',p',x_j^*)\leq P_j(x',p',x_j^*)$.
When $x>x_j^*$,
\begin{eqnarray*}
P_i(x',p',x)=\frac{x_j^*P_i(x',p',x_j^*)}{x}\\
P_j(x',p',x)=\frac{x_j^*P_j(x',p',x_j^*)}{x}
\end{eqnarray*}
Since $P_i(x',p',x_j^*)<P_j(x',p',x_j^*)$ we konw that $P_i(x',p',x)<P_j(x',p',x)$ $\forall x>x_j^*$.
\item [Case 3] $x'>x_j^*>x_i^*$\\
When $x>x_j^*$, both types have equal cost line as $xp=x'p'$. Thus, $P_i(x',p',x)=P_j(x',p',x)$ $\forall x>x_j^*$.
\end{description}
\end{proof}

\end{document}